\documentclass{article}
\usepackage[utf8]{inputenc}

\title{Latroids and code invariants}
\author{}
\author{Elisa Gorla and Flavio Salizzoni}
\date{}

\usepackage{geometry}
\usepackage{amssymb,amsthm,latexsym,amsmath,amscd,graphicx,url,xcolor,comment}
\usepackage[mathscr]{euscript}
\usepackage[bookmarks,
			colorlinks,
			linkcolor=blue,
			citecolor=blue,
            urlcolor=blue,
			linktoc=all]
			{hyperref}	

\theoremstyle{definition}
\newtheorem{theorem}{Theorem}[section]
\newtheorem{proposition}[theorem]{Proposition}
\newtheorem{lemma}[theorem]{Lemma}
\newtheorem{definition}[theorem]{Definition}
\newtheorem{example}[theorem]{Example}
\newtheorem{remark}[theorem]{Remark}
\newtheorem{corollary}[theorem]{Corollary}

\newcommand{\N}{\mathbb N}

\newcommand{\Z}{\mathbb Z}

\newcommand{\F}{\mathbb F}
\newcommand{\K}{\mathbb K}
\newcommand{\Cc}{\mathcal C}
\newcommand{\Dd}{\mathcal D}

\newcommand{\Ll}{\mathcal L}

\newcommand{\wt}{\mathrm{wt}}
\newcommand{\maxwt}{\mathrm{maxwt}}

\newcommand{\supp}{\mathrm{supp}}

\DeclareMathOperator{\minwt}{minwt}
\DeclareMathOperator{\hgt}{ht}

\providecommand{\keywords}[1]
{
  \textbf{\textit{Key words: }} #1
}
\providecommand{\subclas}[1]
{
  \textbf{\textit{Mathematics Subject Classification: }} #1
}

\begin{document}

\maketitle
\begin{abstract}
Latroids were introduced by Vertigan, who associated a latroid to a linear block code and showed that its Tutte polynomial determines the weight enumerator of the code. The original definition of a latroid is in terms of its rank function. For a complemented lattice, we establish cryptomorphic definitions in terms of independent elements, bases, circuits, and flats.
We then associate a latroid to a code over a ring or a field endowed with a general support function and show that the generalized weights of the code can be recovered from the associated latroid. This provides a uniform framework for studying generalized weights and other combinatorial invariants of linear block codes, linear codes over a ring, rank-metric, and sum-rank metric codes. 
\end{abstract}
\tableofcontents

\,

\noindent\keywords{Latroids, Tutte polynomial, weight enumerator, support, linear codes}

\noindent\subclas{94B05, 05B35, 06C20}

\section*{Introduction}

An effective way to study linear codes is to establish connections between them and various algebraic or combinatorial objects that partially capture their structure and their basic properties. Understanding which invariants can be determined from these objects is nowadays a question of great interest within the coding theory community.

In this respect, the most studied case is certainly that of linear block codes. Indeed, since they are finite dimensional vector spaces over finite fields, it is common knowledge that we can associate a matroid to them \cite[Chapter 1]{oxleybook}. In~\cite{greene1976weight} Green showed how the weight enumerator of a linear block code is determined by the Tutte polynomial of the associated matroid. More recently, in~\cite{jurrius2009extended} Jurrius and Pellikaan proved that the Tutte polynomial of the matroid carries the same information as the generalized weight enumerator or the extended weight enumerator, i.e., either one of them determines the others. Britz independently proved a similar result in~\cite{britz2010code}. Starting from the circuits of the matroid, we can also associate a monomial ideal. In~\cite{johnsen2013hamming}, Johnsen and Verdure proved that the generalized Hamming weights of a linear block code are determined by the $\mathbb{N}$-graded Betti numbers of the associated ideal. This relation was further studied in~\cite{garcia2022free,ghorpade2020pure,gorla2022generalized,johnsen2016generalization,johnsen2014stanley,johnsen2021mobius}, among others. 

In~\cite{jurrius2009qmatroid}, Jurrius and Pellikaan showed how to associate a $q$-matroid to a vector rank-metric code. This topic was further investigated in~\cite{gorla2020rank}, where the authors introduced the $q$-analog of a polymatroid and showed how to associate a $q$-polymatroid to a rank-metric code. Among other results, they proved that the generalized rank-weights are determined by the associated $q$-polymatroid. In~\cite{johnsen2023weight}, Johnsen, Pratihar, and Verdure were able to express the generalized rank-weights of a vector rank-metric code in terms of the Betti numbers of a monomial ideal constructed from the associated $q$-matroid. Recently, in~\cite{panja2023some}, Panja, Pratihar, and Randrianarisoa introduced sum-matroids and associated one to each sum-rank metric code.

In~\cite{gorla2022generalized}, Gorla and Ravagnani studied linear codes over rings endowed with a support function. They extended the definition of generalized weights to this setting and proved that generalized weights can be recovered from a monomial ideal associated to the code via the support. However, as we observe in this paper, this ideal often does not contain enough information to recover the weight enumerator of the code. In the case of Hamming support, in~\cite{vertigan2004latroids} Vertigan proved that the weight enumerator of a linear code over a ring is determined by the Tutte polynomial of the associated latroid.  In this paper, we show that one can associate a latroid to a linear code whose support function takes values in a lattice. This includes, e.g., the case of linear block codes, supported on sets, that of rank-metric codes, supported on vector spaces, and that of sum-rank metric codes, supported on products of vector spaces. In addition, notice that the $(\Ll,r)$-polymatroids introduced by Alfarano and Byrne in~\cite{alfarano2023critical} are a special case of $\Z$-latroids. Similar ideas were also developed independently in \cite{pratihar2024lattice}. In general, latroids allow us to treat many families of codes with a unified approach, as they recover matroids, $q$-matroids,  $q$-polymatroids, and sum-matroids as special cases. Because of this, they are a powerful tool for streamlining our approach to the study of linear codes.

The paper is organized as follows.
In Section~\ref{section:prel}, we recall some basic definition on lattices, and the general theory of supports for linear codes over rings from~\cite{gorla2022generalized}. We recall the definition of generalized weights in this context and show that they are invariant under code equivalence. We also give a new, natural definition of generalized weights  for linear codes over rings. As a side remark, we provide a classification of code equivalences. In Section~\ref{section:latroid}, we recall the definition of a latroid and study the basic properties of latroids. In particular, we focus on finite complemented modular lattices, and show that in this situation one can define a latroid using independent sets, bases, circuits, or flats, as in the matroid case. Section~\ref{section:linear} is devoted to linear codes over rings and their associated latroid. In particular, we prove that the weight enumerator of a linear code endowed with the chain support can be determined from the Tutte-Whitney rank generating function of the associated latroid. In the last part of this section, we make some remarks on the associated monomial ideal. Finally, in the last section we discuss how to obtain the $q$-matroid, $q$-polymatroid, or sum-matroid associated to a code from the latroid associated to the same code.

\section{Notation and preliminaries}\label{section:prel}

Throughout the paper, we work over a finite field $\F_q$ of cardinality $q$ or over a finite principal ideal ring $R$.
We denote by $\mathcal{P}(X)$ the power set of a set $X$.

\subsection{Lattices}

A partially ordered set $(\mathcal{L},\leq)$ is called a {\bf lattice} if every pair of elements $L_1,L_2\in\mathcal{L}$ has a least upper bound ({\bf join}), denoted by $L_1\vee L_2$, and a greatest lower bound ({\bf meet}), denoted by $L_1\vee L_2$. A lattice $\mathcal{L}$ is {\bf bounded} if it has a maximum $1_{\Ll}$ and a minimum $0_{\Ll}$. A nonzero element $L\in\Ll$ is called an {\bf atom} if it is minimal. If $L_1\leq L_2\in\Ll$, we say that $L_2$ {\bf dominates} $L_1$. If $L_1\neq L_2$ and there is no $L\in\Ll$ such that $L_1< L< L_2$, we say that $L_2$ {\bf covers} $L_1$.
We denote by $[L_1,L_2]$ the {\bf interval} between $L_1$ and $L_2$, that is, the sublattice $\{L\in\Ll:L_1\leq L\leq L_2\}$ of $\Ll$. 

We are mainly interested in modular lattices.

\begin{definition}
A lattice $\mathcal{L}$ is called {\bf modular} if for every $L_1,L_2,L_3\in\mathcal{L}$ with $L_1\leq L_3$ we have 
$$L_1\vee(L_2\wedge L_3)= (L_1\vee L_2)\wedge L_3.$$
\end{definition}

The set of normal subgroups of a group, the set of subspaces of a vector space, and the set of submodules of a module are examples of modular lattices. When the operations of join and meet distribute over each other, the lattice is called {\bf distributive}. Every distributive lattice is modular. A typical example of a distributive lattice is a power set with union and intersection as join and meet.

\begin{definition}
A {\bf complemented lattice} $\Ll$ is a bounded lattice in which for every $L_1\in\Ll$ there exists $L_2\in\Ll$ such that 
$$L_1\wedge L_2=0_{\Ll}\text{ and }L_1\vee L_2=1_{\Ll}.$$
A lattice $\mathcal{L}$ is called {\bf relatively complemented} if every interval in $\Ll$ is complemented, i.e., for every $L_1,L_2,L_3\in\mathcal{L}$ such that $L_1\leq L_2\leq L_3$  there exists $\bar L\in\mathcal{L}$ such that 
$$L_2\wedge \bar L=L_1\text{ and }L_2\vee \bar L=L_3.$$
\end{definition}

An example of a complemented lattice is the set of vector subspaces of $K^n$ ordered by inclusion, where $K$ is a field. If $R$ is a finite principal ideal ring that is not isomorphic to a product of fields, e.g. if $R$ is a finite chain ring that is not a field, then the lattice of ideals of $R$ ordered by inclusion and the lattice of submodules of $R^n$ ordered by inclusion are not complemented.

A complemented lattice is relatively complemented if it is modular.

\begin{definition}
A bounded lattice $\mathcal{L}$ is {\bf graded} if there exists a function $\hgt:\mathcal{L}\rightarrow \Z$ such that $\hgt(0_{\mathcal{L}})=0$ and $\hgt(L)+1=\hgt(M)$ for all $L,M\in\mathcal{L}$ such that $M$ covers $L$. This function is unique and is called the {\bf height function} of $\mathcal{L}$.
\end{definition}

Recall that a finite lattice is modular if and only if it is graded and its height function is {\bf modular}, i.e., it satisfies $\hgt(L)+\hgt(M)=\hgt(L\vee M)+\hgt(L\wedge M)$ for every $L,M\in\Ll$.  The following classical result gives a complete classification of finite complemented modular lattices. It may help to understand in what generality the results of Subsection~\ref{section:crypto} apply.

\begin{theorem}[{\cite[Theorem 7.56]{cameron2008introduction}}]
    Let $\Ll$ be a finite complemented modular lattice. Then, $\Ll$ is the direct product of a finite number of lattices of the form:
    \begin{enumerate}
        \item $\Ll'=\{0,1\}$,
        \item a proper line,
        \item a proper projective plane,
        \item a subspace lattice of a finite dimensional vector space over a finite field.
    \end{enumerate}
\end{theorem}

\begin{definition}
A finite lattice $\mathcal{L}$ is {\bf atomistic} if every $L\in\Ll\setminus \{0_{\Ll}\}$ is a join of atoms. It is {\bf semimodular} if, for $L_1,L_2\in\Ll$ such that $L_1$ and $L_2$ cover $L_1\wedge L_2$, $L_1\vee L_2$ covers $L_1$ and $L_2$. A finite lattice is {\bf geometric} if it is atomistic and semimodular.
\end{definition}

We conclude this subsection with the definition of an ordered abelian group.

\begin{definition}
    An {\bf ordered abelian group} is a triple $(A,+,\leq)$, where $(A,+)$ is an abelian group and $\leq$ is a partial order on $A$ such that for all $a_1,a_2,a_3\in A$ $a_1\leq a_2$ implies $a_1+a_3\leq a_2+a_3$. In particular, we have the following:
    \begin{enumerate}
        \item $a_1\leq a_2$ if and only if $0\leq a_2-a_1$,
        \item if $a_1,a_2\geq 0$, then $a_1+a_2\geq 0$.
    \end{enumerate}
\end{definition}

\begin{example}
    We are mainly interested in the ordered abelian group\index{ordered group} $(\Z^u,+,\leq)$ with $u\in\N$, where $(a_1,\dots, a_u)\leq(a_1',\dots a_u')$ if and only if $a_i\leq a_i'$ for $i\in[u]$ in the usual order on $\Z$. 
\end{example}

\subsection{Supports of $R$-linear codes}

Even though many definitions and results that we will discuss throughout the paper hold for infinite commutative rings, for the sake of simplicity, we restrict ourselves to finite rings. In the sequel, we let $R$ be a finite unitary commutative ring. For a finitely generated $R$-module $M$, we denote by $\lambda(M)$ its length and by $\mu(M)$ the least cardinality of a (minimal) system of generators of $M$. By convention, we have $\mu(0)=0$. 

\begin{definition}
An {\bf $R$-linear code} $\Cc$ is an $R$-submodule of $R^n$. 
\end{definition}

Gorla and Ravagnani formulated a general theory of supports over rings in~\cite{gorla2022generalized}. Here we recall what is needed for our purposes.

\begin{definition}\label{definition:ringsupport}
    A {\bf support\index{support!ring}} on $R^n$ is a function $\supp:R^n\rightarrow \Z^u$ such that:
    \begin{enumerate}
        \item $\supp(v)=0$ if and only if $v=0$.
        \item $\supp(rv)\leq\supp (v)$ for all $r\in R$ and $v\in R^n.$
        \item $\supp(v+w)\leq \supp(v)\vee\supp(w)$ for all $v,w\in R^n$ 
    \end{enumerate}
    A support is called {\bf modular\index{support!modular}} if it satisfies the following property:
    \begin{enumerate}
        \item[4.] If $v,w\in R^n$ and $i\in[u]$ satisfy $\supp (v)_i\leq \supp(w)_i$, then there exists $r\in R$ such that $\supp(v+rw)_i<\supp(v)_i$.
    \end{enumerate}
\end{definition}

Notice that the Hamming support is modular. An example of a support on $\F_q^n$ that is not modular for $n\geq 2$ is given by the function $\tau:\F_q^n\rightarrow \Z$ that maps the zero vector to $0$ and every other vector to $1$.

A support $\supp:R^n\rightarrow \Z^u$ naturally induces a function $\supp:\mathcal{P}(R^n)\rightarrow\Z^u$, defined as ${\supp}(X)=\bigvee_{x\in X}\supp(x)$ for $X\in \mathcal{P}(R^n)$. Due to the properties of the support, if $X=\langle x_1,\ldots,x_\ell\rangle$, then ${\supp}(X)=\supp(x_1)\vee\ldots\vee\supp(x_\ell)$

In coding theory, the notion of support is closely related to that of weight. The Hamming support, for example, gives rise to the Hamming weight on $\F_q^n$. 

\begin{definition}\label{definition:weightring}
The {\bf weight\index{weight!ring}} of $v\in R^n$ with respect to $\supp$ is the $1$-norm of the support of $v$, i.e. $\wt(v)=\lvert\supp(v)\rvert$. The weight of an $R$-linear code $\Cc$ is defined as $\wt(\Cc)=\lvert \supp(\Cc)\rvert$. The {\bf minimum} and the {\bf maximum weight} of a code $0\neq\Cc\subseteq R^n$ are, respectively,
$$\minwt(\Cc)=\min\left\{\wt(v):v\in\Cc\setminus 0\right\}\text{ and }\maxwt(\Cc)=\max\left\{\wt(v):v\in\Cc\right\}.$$
\end{definition} 

One can easily see that the weight defined above is an invariant weight function, but it is not always homogeneous. We refer to~\cite[Section 2]{greferath2013characteristics} for the relevant definitions.

Notice that there exist weights of interest to the coding theory community, whose corresponding ``support" does not satisfy the condition of Definition~\ref{definition:ringsupport}. For example, the support $\supp_{\mathrm{L}}:\Z_4\rightarrow \Z$, associated to the Lee weight $\wt_{\mathrm{L}}:\Z_4\rightarrow \Z$, is given by $\supp_{\mathrm{L}}(0)=0$, $\supp_{\mathrm{L}}(1)=\supp_{\mathrm{L}}(3)=1$, and $\supp_{\mathrm{L}}(2)=2$. This is not a support according to Definition~\ref{definition:ringsupport}. In fact, it does not satisfy the second condition, since $\supp_{\mathrm{L}}(2)=\supp_{\mathrm{L}}(2\cdot 1)>\supp_{\mathrm{L}}(1)$. 

Recall that if $R$ is a finite ring, then there exist $R_1,\dots,R_{\ell}$ finite local rings such that $R\cong R_1\times\dots\times R_{\ell}$, see~\cite[Theorem 8.7]{atiyah2018introduction}. In particular, if $R$ is a principal ideal ring, then $R_1,\dots,R_{\ell}$ are also principal ideal rings. 
Abusing the notation, from now on we will write $R=R_1\times\dots\times R_n$. Similarly, we will write $R^n=R_1^n\times\dots\times R_{\ell}^n$ and $\Cc=\Cc_1\times\dots\times\Cc_{\ell}$, respectively, instead of $R^n\cong R_1^n\times\dots\times R_{\ell}^n$ and $\Cc\cong\Cc_1\times\dots\times\Cc_{\ell}$.
A finite local commutative principal ideal ring is often called a {\bf finite chain ring\index{finite chain ring}}. If $R$ is a finite chain ring, then any element $r\in R$ is of the form $r=\upsilon\alpha^k$, where $\upsilon$ is an invertible element and $\alpha$ is a generator of the maximal ideal of $R$. 
The next result allows us to reduce the study of supports of rings to that of supports of local rings.

\begin{proposition}[{\cite[Theorem 2.23]{gorla2022generalized}}]\label{proposition:supportsplits}
Let $\supp:R^n\rightarrow \Z^u$ be a modular support. Up to a permutation of the coordinates of $\Z^u$ we have that $\supp=\supp_1\times\dots\times\supp_{\ell}$, where $\supp_i:R^{n}_i\rightarrow \Z^{u_i}$ for $i\in[\ell]$ and $u_i\in\N$ with $u_1+\dots+u_{\ell}=u$. Moreover, $\supp_i$ is a modular support for all $i\in[\ell]$.
\end{proposition}

Let $\supp_1:R^{n_1}\rightarrow \Z^{u_1}$ and $\supp_2:R^{n_2}\rightarrow \Z^{u_2}$ be supports. It is easy to see that the product $\supp_1\times\supp_2$ is a  support from $R^{n_1+n_2}$ to $\Z^{u_1+u_2}$. Moreover, the product is modular whenever the factors are. A support is called standard if it can be decomposed as a product of supports, each defined on a single copy of $R$.

\begin{definition}
A support $\supp:R^n\rightarrow \Z^u$ is {\bf standard\index{support!standard}} if for each $i\in[n]$ there exist $u_i\in\N$ and a support $\supp_i:R\rightarrow \Z^{u_i}$ such that $\supp(r_1,\dots,r_n)=(\supp_1(r_1),\dots,\supp_n(r_n))$, up to permuting the coordinates of $\Z^u$.    
\end{definition}

Notice that, for a standard support $\supp$, one has
$$\supp(r_1,\dots,r_n)=\supp(r_1,0,\dots,0)\vee \dots\vee\supp(0,\dots,0,r_n).$$
We will be interested in a specific standard modular support for finite chain rings introduced in~\cite[Example 26]{ravagnani2018duality} and defined as follows.

\begin{definition}\label{definition:chainsupport}
Let $R$ be a finite chain ring with maximal ideal $(\alpha)$. Let $k$ be the smallest positive integer such that $\alpha^k=0$. Let $\overline{\supp}:R\rightarrow \Z$ be the support function given by $$\overline{\supp}(r)=\min\left\{0\leq i\leq k: r\in\left(\alpha^{k-i}\right)\right\}=\lambda(r),$$
for every $r\in R$, where $\lambda(r)$ is the length of the ideal generated by $r$. The support $\supp=\overline{\supp}\times\dots\times\overline{\supp}:R^n\rightarrow\Z^n$ is called the {\bf chain support\index{support!chain}} on $R^n$.
\end{definition}

\subsection{Generalized weights}

Let $\Cc\subseteq R^n$ be an $R$-linear code. Since $R=R_1\times\dots\times R_{\ell}$ with $R_i$ finite local ring for all $i\in[\ell]$, we have that $\Cc=\Cc_1\times\dots\times\Cc_{\ell}$,  where $\Cc_i\subseteq R_i^n$ is the projection $\pi_i(\Cc)$ of $\Cc$ on the $i$-th factor of $R^n=R_1^n\times\dots\times R_{\ell}^n$ for all $i\in[\ell]$. Recall that $\mu(\Cc)$ denotes the least cardinality of a system of generators of a code $\Cc$. For a code $\Cc=\Cc_1\times\dots\times\Cc_{\ell}\subseteq R^n$, we set $M(\Cc):=\mu(\Cc_1)+\dots+\mu(\Cc_{\ell})$. We now have all the necessary elements to state the definition of generalized weights of an $R$-linear code with respect to a support $\supp$, as was given in~\cite{gorla2022generalized}.

\begin{definition}\label{definition:genweiring}
    For $r\in [M(\Cc)]$, the {\bf $r$-th generalized weight\index{generalized weights!rings}} of $\Cc$ is given by
    $$d_r(\Cc)=\min\left\{\wt(\Dd):\Dd\in S_j(\Cc)\text{ for some }j\geq r\right\},$$
    where $S_j(\Cc)=\{\Dd\subseteq\Cc: \Dd\text{ is a subcode with a minimal system of generators of cardinality }j\}$.
\end{definition}

Notice that the previous definition is well-posed, since $S_j(\Cc)\neq\emptyset$ for $j\in[M(\Cc)]$ as proved in~\cite[Theorem 1.8]{gorla2022generalized}. When $R$ is a finite field, the cardinality of a minimal system of generators coincides with the dimension of the subcode that they generate. Therefore, Definition~\ref{definition:genweiring} extends the classical definition of generalized weights of linear block codes. When $R$ is a ring, however, a code may have minimal systems of generators of different cardinalities, see e.g.~\cite[Example 1.6]{gorla2022generalized}. The next proposition collects some basic properties of generalized weights.

\begin{proposition}[{\cite[Lemma 2.12]{gorla2022generalized}}]\label{proposition:basicpropgwr}
    Let $D\subseteq C\subseteq R^n$ be two $R$-linear codes. Then
    \begin{enumerate}
        \item $d_1(\Cc)=\minwt(\Cc)$,
        \item $d_r(\Dd)\geq d_r(\Cc)$ for $r\in[\min\{M(\Dd),M(\Cc)\}]$,
        \item $d_{r+1}(\Cc)\geq d_{r}(\Cc)$ for $r\in[M(\Cc)-1]$,
        \item $d_r(\Cc)=\min\{\lvert\supp(\Dd)\rvert:D\subseteq\Cc\text{ and }M(\Dd)\geq r\}$ for $r\in [M(\Cc)]$.
    \end{enumerate}
\end{proposition}

One of the reasons for the interest in generalized Hamming weights is that they are invariant under code-equivalence. Here, we prove that this is also the case for $R$-linear codes. We start by defining equivalence between $R$-linear codes.

\begin{definition}\label{definition:equivalencesrings}
    An {\bf isometry\index{isometry}} between $R$-linear codes is an $R$-module isomorphism $\varphi:\Cc_1\rightarrow \Cc_2$ that preserves the weight, i.e., such that $\wt(v)=\wt(\varphi(v))$ for all $v\in \Cc_1$.
    Two $R$-linear codes $\Cc_1$ and $\Cc_2$ in $R^n$ are {\bf equivalent\index{equivalence!rings}} if there exists an isometry $\varphi:R^n\rightarrow R^n$ that maps $\Cc_1$ to $\Cc_2$.
\end{definition}

A classical result for the Hamming support states that an isometry from $\F_q^n$ to itself can be expressed as multiplication by a permutation matrix and a diagonal one. In the following, we establish a similar result for codes over principal ideal rings equipped with a standard modular support. We start by considering the case where $R$ is a finite chain ring.

\begin{lemma}\label{lemma:classificationequivrings}
Let $R$ be a finite chain ring, let $\supp=\supp_1\times\dots\times\supp_n$ be a standard modular support on $R^n$, and let $\varphi:R^n\rightarrow R^n$ be an isometry with respect to the weight induced by $\supp$. Then there exist a diagonal invertible matrix $D$ and a permutation matrix $M$ such that $\varphi(v)=DMv$ for all $v\in R^n$.
\end{lemma}

\begin{proof}
It is known that an $R$-module isomorphism from $R^n$ to itself can be expressed as multiplication by a matrix $N=(n_{i,j})$ in $R^{n\times n}$. In order to prove the statement, we proceed by induction on $n$. When $n=1$, the thesis is clear. Assume that the statement holds for $n-1$. Without loss of generality, assume that $ \lvert\supp_1(1)\rvert\leq\dots\leq\lvert\supp_n(1)\rvert$. Let $e_i$ be an element of the standard basis. Then, an entry of $\varphi(e_i)$ must be invertible, otherwise $\varphi$ would not be injective. We start by considering $e_1$. Since $\lvert\supp_1(1)\rvert\leq \lvert\supp_i(1)\rvert$ for $i>1$, we conclude that the first column of $N$ has an invertible entry, say the $k$-th entry, and it is zero everywhere else. Up to multiplying by a permutation matrix, we may assume $k=1$. 
Consider the vector $v=(-n_{1,2}, n_{1,1},0,\dots,0)^t\in R^n$. Since $\varphi(v)=n_{1,1}(0,n_{2,2},\ldots,n_{t,2})$, then $\wt(\varphi(v))\leq \wt(\varphi(e_2))$, while $\wt(v)\geq\wt(e_2)$. Since $\varphi$ is an isometry, we have $n_{1,2}=0$. Proceeding in a similar way, we show that the first row of $N$ is different from zero only in the first entry. This implies that $\varphi$ restricted to $0\times R^{n-1}$ is an isometry of $R^{n-1}$. We conclude by the induction hypothesis. 
\end{proof}

When $R$ is a principal ideal ring, the isometries of $R^n$ can still be expressed as a product by a matrix, but the matrices that represent the isometries have a more complicated description. However, we can classify the isometries of $R^n$ based on the isometries of finite chain rings that we described in the previous lemma. 

\begin{theorem}\label{theorem:classificationequivrings}
Let $R=R_1\times\dots\times R_{\ell}$ be a principal ideal ring, let $\supp$ be a standard modular support on $R^n$, and let $\varphi:R^n\rightarrow R^n$ be an isometry with respect to $\supp$. Then, for each $i\in[\ell]$, there exists an isometry $\varphi_i:R_i^n\rightarrow R_i^n$ such that $\pi_i(\varphi(r)))=\varphi_i(\pi_i(r))$ for every $r\in R^n$, where $\pi_i:R^n\rightarrow R_i^n$ is the standard projection. Equivalently, the matrix defining $\varphi$ has the form $DM$, where $D$ is invertible and diagonal, and $\pi_i(M)$ is a permutation matrix for all $i\in[\ell]$.
\end{theorem}

\begin{proof}
This follows from observing that any $R$-module isomorphism maps $0\times\dots\times R_i^n\times\dots\times 0$ to itself and that the restriction of an isometry is an isometry.
\end{proof}

\begin{example}
Over the ring $\Z_6=\Z_2\times \Z_3$ let $\supp:\Z_6\rightarrow \Z^2$ be given by $\supp(1)=(1,1)$, $\supp(2)=(1,0)$ and $\supp(3)=(0,1)$, i.e., $\supp=\supp_1\times\supp_2$, where $\supp_1$ is the Hamming support on $\Z_3$ and $\supp_2$ is the Hamming support on $\Z_2$. 
Consider the free module $\Z_6^2$  with the standard modular support $\supp\times \supp$. One can check by direct computation that multiplication by the matrix $$N=\begin{pmatrix} 2&3\\3&2
\end{pmatrix}\in\Z_6^{2\times 2}$$
is an isometry $\varphi:\Z_6^2\rightarrow\Z_6^2$. Clearly $N$ is not the product of a permutation matrix and a diagonal one. However, if we look at the projections in $\Z_2^2$ and $\Z_3^2$, we find that the two isometries $\varphi_1$ and $\varphi_2$ of the proof of Theorem~\ref{theorem:classificationequivrings} correspond, respectively, to the matrices
$$N_1=\begin{pmatrix} 0&1\\1&0
\end{pmatrix}\in\Z_2^{2\times 2}\;\;\;\text{ and }\;\;\;N_2=\begin{pmatrix} 2&0\\0&2
\end{pmatrix}\in\Z_3^{2\times 2},$$
that are both permutation matrices multiplied by a diagonal invertible matrix, as required by Lemma~\ref{lemma:classificationequivrings}. Equivalently, $N=DM\in\Z_6^{2\times 2}$ with $$D=\begin{pmatrix} 5&0\\0&5
\end{pmatrix}\;\;\; \mbox{ and }\;\;\; M=\begin{pmatrix} 4&3\\3&4
\end{pmatrix}.$$
\end{example}

Theorem~\ref{theorem:classificationequivrings} implies that the generalized weights are a family of invariants.

\begin{corollary}\label{corollary:invariantring}
    The generalized weights of an $R$-linear code are invariant under equivalences.
\end{corollary}

\begin{proof}
Let $\varphi:R^n\rightarrow R^n$ be an equivalence between two $R$-linear codes $\Cc_1$ and $\Cc_2$. Consider a minimal system of generators $M$ of a subcode $\Dd_1$ of $\Cc_1$. Since $\varphi$ is an $R$-linear isomorphism of $R^n$, $\varphi(M)$ is a minimal system of generators of a subcode $\Dd_2$ of $\Cc_2$. In particular, $M(\Cc_1)=M(\Cc_2)$. Since $\varphi$ is an isometry, then $\wt(\Dd_1)=\wt(\Dd_2)$. Therefore, $d_r(\Cc_1)\geq d_r(\Cc_2)$ for $r\in [M(\Cc_2)]$. This suffices to conclude, since $\varphi^{-1}$ is also an isometry.
\end{proof}

The dimension of a vector space coincides with the cardinality of a minimal system of generators. This is not always the case for a module over a ring. For this reason, the definition of generalized weights does not extend uniquely to the case of modules on rings: Definition~\ref{definition:genweiring} is only one possible choice. A different choice is made, e.g., in~\cite{kamche2019rank,blanco2025macwilliams,martinez-penas2024maximum}. In this paper, we explore the following definition, which is related to the approach of~\cite{gorla2018algebraic,gorla2025matrix}.

\begin{definition}\label{definition:genweiring2}
    For $1\leq r\leq \lambda(\Cc)$, the {\bf $r$-th generalized weight} of $\Cc$ is 
    $$\bar d_r(\Cc)=\min\left\{\wt(\Dd):\Dd\text{ is a submodule of }\Cc\text{ such that }\lambda(\Dd)\geq r\right\},$$
    for $1\leq r\leq \lambda(\Cc)$.
\end{definition}

It is easy to prove that these generalized weights satisfy properties similar to those of Proposition~\ref{proposition:basicpropgwr}.

\begin{proposition}\label{proposition:basicpropgwr2}
Let $R$ be endowed with a support and let $\Dd\subseteq\Cc\subseteq R^n$ be $R$-linear codes. Then
    \begin{enumerate}
        \item $\bar d_1(\Cc)=\minwt(\Cc)$,
        \item $\bar d_r(\Dd)\geq \bar d_r(\Cc)$ for $r\in[ \lambda(\Dd)]$,
        \item $\bar d_{r+1}(\Cc)\geq\bar d_{r}(\Cc)$ for $r\in[\lambda(\Cc)-1]$,
        \item if the support is modular, then $\bar d_{r+1}(\Cc)>\bar d_{r}(\Cc)$ for $r\in[\lambda(\Cc)-1]$,
        \item $\bar d_r(\Cc)$ is invariant under equivalence for all $r\in[\lambda(\Cc)]$.
    \end{enumerate}
\end{proposition}

\begin{proof}
    Items 1, 2, 3, and 5 follow from the definition of generalized weights and the properties of supports. In order to prove item 4, let $\Dd$ be a submodule of $\Cc$ that realizes $\bar d_r(\Cc)$. Let $d\in\Dd$ be such that there exists an index $i$ for which $\supp(d)_i\geq\supp(c)_i$ for all $c\in\Dd$. The set $\bar \Dd=\{c\in \Dd:\supp(c)_i<\supp(d)_i\}$ is a submodule of $\Dd$. We claim that $\lambda(\bar\Dd)=\lambda(\Dd)-1$. In fact, assume that there exists a module $\tilde \Dd$ such that $\bar\Dd\subsetneq\tilde\Dd\subseteq\Dd$. Then, there exists $b\in\tilde\Dd$ such that $\supp(b)_i=\supp(d)_i$. 
    Let $a\in\Dd$ be such that $\supp(a)_i=\supp(d)_i$. 
    Since the support is modular, there exists $r\in R$ such that $\supp(a+rb)_i<\supp(a)_i=\supp(d)_i$. Therefore, $a+rb\in\bar\Dd$ and so $a\in\tilde\Dd$. This implies $\tilde \Dd=\Dd$, and hence $\lambda(\bar\Dd)=\lambda(\Dd)-1$. Finally, by construction $\wt(\bar\Dd)<\wt(\Dd)$, so we conclude that $\bar d_{r-1}(\Cc)<\bar d_r(\Cc)$.
\end{proof}

\section{Latroids}\label{section:latroid}

\subsection{Definition and basic properties}

Latroids were introduced in~\cite{vertigan2004latroids} by Vertigan as a generalization of matroids. It is easy to see that they also generalize their $q$-analogs. Since our goal is to associate a latroid to a linear code, we only give the definition for finite lattices. However, some of the results can be extended with due attention to the infinite case.  

\begin{definition}
Let $A$ be an ordered abelian group and let $\mathcal{L}$ be a finite modular lattice. An {\bf $A$-latroid}\index{latroid} with {\bf rank function} $\rho:\mathcal{L}\rightarrow A$ under {\bf length function} $\lVert \cdot\rVert:\mathcal{L}\rightarrow A$ on the lattice $\mathcal{L}$, is a triple $(\rho,\lVert \cdot\rVert, \mathcal{L})$ such that:
\begin{enumerate}
\item[L1.] $\rho(0_{\mathcal{L}})=\lVert 0_{\mathcal{L}}\rVert=0_A$,
\item[L2.] $\lVert \cdot \rVert$ is strictly increasing, that is, $\lVert L\rVert<\lVert M\rVert$ for all $L,M\in \mathcal{L}$ with $L<M$,
\item[L3.] $\lVert \cdot \rVert$ is {\bf modular}, that is, $\lVert L\rVert+\lVert M\rVert=\lVert L\vee M\rVert+\lVert L\wedge M\rVert$ for all $L,M\in \mathcal{L}$,
\item[L4.] $\rho$ is {\bf bounded increasing}, that is, $0\leq \rho(M)-\rho(L)\leq \lVert M\rVert-\lVert L\rVert$ for all $L,M\in \mathcal{L}$ with $L\leq M$,
\item[L5.] $\rho$ is {\bf submodular}, that is, $\rho(L)+\rho(M)\geq \rho(L\vee M)+\rho(L\wedge M)$ for all $L,M\in \mathcal{L}$.
\end{enumerate}
\end{definition}

\begin{example}
    Let $\Ll$ be a finite graded modular lattice with height function $\hgt_{\Ll}$. Then, the triple $(\hgt_{\Ll},\hgt_{\Ll},\Ll)$ is a $\Z$-latroid. Every latroid of the form $(\lVert\cdot\rVert,\lVert\cdot\rVert,\Ll)$ is called a {\bf free latroid}.
\end{example}

\begin{example}
    Let $\Ll$ be a finite modular lattice and let $\lVert\cdot\rVert:\Ll\rightarrow A$ be a modular function with $\lVert 0_{\mathcal{L}}\rVert=0_A$. For $0<a\in A$, let $\rho_a:\Ll\rightarrow A$ be the function defined by
    \begin{equation*}
        \rho_a(L)=\begin{cases}
            \lVert L\rVert&\text{if } \lVert L\rVert\leq a\\
            a&\text{otherwise}
        \end{cases}
    \end{equation*}
    for all $L\in\Ll$. We claim that $(\rho_a,\lVert\cdot\rVert,\Ll)$ is an $A$-latroid. Indeed, L1, L2, L3, and L4 are trivially satisfied. For every $L,M\in\Ll$, if $\rho_a(L)=a$, then also $\rho_a(L\vee M)=a$, and so $\rho_a(L)+\rho_a(M)\geq \rho_a(L\vee M)+\rho_a(L\wedge M)$. Otherwise, if $\rho_a(L),\rho_a(M)<a$, then
    $$\rho_a(L)+\rho_a(M)=\lVert L\rVert+\lVert M\rVert=\lVert L\vee M\rVert+\lVert L\wedge M\rVert\geq \rho_a(L\vee M)+\rho_a(L\wedge M).$$
    So L5 is also satisfied. In analogy to matroid theory terminology, we call $(\rho_a,\lVert\cdot\rVert,\Ll)$ a {\bf uniform latroid}.
\end{example}

\begin{definition}
Consider an $A$-latroid $(\rho,\lVert\cdot\rVert,\Ll)$ and let $L_1\leq L_2\in\Ll$. Define $\lVert\cdot\rVert_{[L_1,L_2]}:[L_1,L_2]\rightarrow A$ and $\rho_{[L_1,L_2]}:[L_1,L_2]\rightarrow A$ as
\begin{equation*}
    \lVert L\rVert_{[L_1,L_2]}=\lVert L\rVert-\lVert L_1\rVert\;\;\text{ and }\;\;\rho_{[L_1,L_2]}(L)=\rho(L)-\rho(L_1),
\end{equation*}
for $L\in[L_1,L_2]$. 
\end{definition}

Notice that $(\rho_{[L_1,L_2]},\lVert\cdot\rVert_{[L_1,L_2]},[L_1,L_2])$ is an $A$-latroid by~\cite[Lemma 5.10]{vertigan2004latroids}. 

\begin{definition}
The {\bf direct sum} $(\rho_1,\lVert\cdot\rVert_1,\Ll_1)\oplus (\rho_2,\lVert\cdot\rVert_2,\Ll_2)$ of $A$-latroids $(\rho_1,\lVert\cdot\rVert_1,\Ll_1)$ and $(\rho_2,\lVert\cdot\rVert_2,\Ll_2)$ is the $A$-latroid $(\rho,\lVert\cdot\rVert,\Ll_1\times\Ll_2)$, where $\rho:\Ll_1\times \Ll_2\rightarrow A$ and $\lVert\cdot\rVert:\Ll_1\times \Ll_2\rightarrow A$ are given by
\begin{equation*}
    \rho(L_1,L_2)=\rho_1(L_1)+\rho_2(L_2)\text{ and }\lVert(L_1,L_2)\rVert=\lVert L_1\rVert_1+\lVert L_2\rVert_2.
\end{equation*}
\end{definition}

To prove that the direct sum of two latroids is a latroid, it suffices to observe that $(L_1,L_2)\vee(M_1,M_2)=(L_1\vee M_1,L_2\vee M_2)$ and $(L_1,L_2)\wedge(M_1,M_2)=(L_1\wedge M_1,L_2\wedge M_2)$ for every $L_1,M_1\in \Ll_1$ and $L_2,M_2\in\Ll_2$. Therefore,
\begin{equation*}
\begin{split}
    &\rho(L_1,L_2)+\rho(M_1,M_2)=\rho_1(L_1)+\rho_2(L_2)+\rho_1(M_1)+\rho_2(M_2)=\\
    &\geq \rho_1(L_1\vee M_1)+\rho_1(L_1\wedge M_1)+\rho_2(L_2\vee M_2)+\rho_2(L_2\wedge M_2)=\\&=\rho(L_1\vee M_1,L_2\vee M_2)+\rho(L_1\wedge M_1,L_2\wedge M_2),
\end{split} 
\end{equation*}
which proves the submodularity of $\rho$. Properties L1, L2, L3, and L4 can be proved in a similar way.

We denote by $(\Ll^{\perp},\leq)=(\{L^{\perp}:L\in\Ll\},\leq)$ the {\bf dual lattice}  of a lattice $\Ll$, where $L_1^{\perp}\leq L_2^{\perp}$ if and only if $L_2\leq L_1$. The dual of a latroid was introduced in \cite[Definition 5.13]{vertigan2004latroids}.

\begin{definition}
    The {\bf dual} of an $A$-latroid $(\rho,\lVert\cdot\rVert,\Ll)$, denoted by $(\rho,\lVert\cdot\rVert,\Ll)^\perp$, is the $A$-latroid $(\rho^{\perp},\lVert\cdot\rVert^{\perp},\Ll^{\perp})$,
where 
\begin{enumerate}
    \item $\lVert L^{\perp}\rVert^{\perp}=\lVert 1_{\mathcal{L}}\rVert-\lVert L\rVert$,
    \item $\rho^{\perp}(L^{\perp})=\lVert L^{\perp}\rVert^{\perp}-\rho(1_{\mathcal{L}})+\rho(L)$,
    for every $L^{\perp}\in\Ll^{\perp}$.
\end{enumerate}
\end{definition}

The next lemma collects some basic properties of the dual latroid.

\begin{lemma}[{\cite[Lemma 5.14]{vertigan2004latroids}}]
    Let $(\rho,\lVert\cdot\rVert,\Ll)$ be an $A$-latroid. Then:
    \begin{enumerate}
        \item $(\rho^{\perp},\lVert\cdot\rVert^{\perp},\Ll^{\perp})$ is an $A$-latroid,
        \item $(\rho^{\perp},\lVert\cdot\rVert^{\perp},\Ll^{\perp})^{\perp}=(\rho,\lVert\cdot\rVert,\Ll),$
        \item for $L_1\leq L_2\in\Ll$ we have $(\rho_{[L_1,L_2]})^{\perp}=\rho^{\perp}_{[L_2^{\perp},L_1^{\perp}]}$,
        \item for $L_1\leq L_2\in\Ll$ we have $(\lVert\cdot\rVert_{[L_1,L_2]})^{\perp}=\lVert\cdot\rVert^{\perp}_{[L_2^{\perp},L_1^{\perp}]}.$
    \end{enumerate}
\end{lemma}

In the next remark we clarify how the concept of latroid generalizes the one of matroid.

\begin{remark}\label{remark:latroidtomatroid}
Let $E$ be a finite set. The power set $\mathcal{P}(E)$ of $E$ is a complete lattice with respect to the union and intersection. It is easy to verify that the cardinality function $\lvert\cdot\rvert$ is a strictly increasing modular function on $\mathcal{P}(E)$. Let $\rho: \mathcal{P}(E)\rightarrow \mathbb{Z}$ be any function for which $(\rho,\lvert\cdot\rvert,\mathcal{P}(E))$ is a $\Z$-latroid. Then $\{X\subseteq E: \lvert X\rvert-\rho(X)>0\text{ and }X\text{ is minimal with this property}\}$ is the set of circuits of a matroid with ground set $E$ and rank function $\rho$. This is a direct consequence of~\cite[Proposition 11.1.1]{oxleybook}.

Conversely, let $(E,\rho)$ be a matroid with ground set $E$ and rank function $\rho$. Then $(\rho,\lvert\cdot\rvert,\mathcal{P}(E))$ is a $\Z$-latroid. In fact, from the axioms of matroids we directly obtain that $\rho(0_{\mathcal{L}})=0$, $\rho$ is submodular, and $0\leq \rho(M)-\rho(L)$ for all $L,M\in\mathcal{P}(E)$ with $L\leq M$. It remains to prove that
$\rho(M)-\rho(L)\leq \lvert M\rvert-\lvert L\rvert$. By the submodularity of $\rho$, we have $\rho(M)\leq\rho(M\setminus L)+\rho(L)$, and by the modularity of the cardinality, we conclude that
$$\rho(M)-\rho(L)\leq\rho(M\setminus L)\leq\lvert M\setminus L\rvert=\lvert M\rvert-\lvert L\rvert.$$
\end{remark}

\subsection{Cryptomorphic definitions}\label{section:crypto}

Inspired by matroid theory, we define and study the concepts of independent element, basis, and circuit of a latroid.

\begin{definition}
    Let $(\rho,\lVert \cdot\rVert, \mathcal{L})$ be an $A$-latroid. An element $L\in\Ll$ is called {\bf independent} if $\rho(L)=\lVert L\rVert$, a {\bf basis} if $\rho(L)=\lVert L\rVert=\rho(1_{\Ll})$, and a {\bf circuit} if $\rho(L)<\lVert L\rVert$ and for all $M<L$ we have $\rho(M)=\lVert M\rVert$.
\end{definition}

In the case of matroids, knowing the independent sets is equivalent to knowing the rank function. However, in the case of latroids, this is not true in general, as one can see in the next example. In Proposition~\ref{proposition:independent}, however, we show that a $\mathbb{Z}$-latroid over a finite complemented modular lattice is determined by its independent elements.

\begin{example}\label{ex:latroid&indep}
Let $\mathcal{L}$ be the lattice of ideals of $\Z_8$. We consider the length function $\lVert\cdot\rVert$ given by $\lVert R\rVert=\vert R\rvert-1$ for each $R\in\Ll$. Then, $(\lfloor\vert\cdot\vert/2\rfloor,\lVert\cdot\rVert,\Ll)$ and  $(\lambda,\lVert\cdot\rVert,\Ll)$ are two $\Z$-latroids whose independent elements are $0_{\Ll}$ and $4\Z_8$. However, we have that $\lfloor\vert\Z_8\vert/2\rfloor=4$, while $\lambda(\Z_8)=3$.
\end{example}

Notice that, if $\lVert\cdot\rVert$ is the length function of a lattice $\Ll$, then any multiple of it by a positive integer is also a length function on $\Ll$. In the example above, the lattice is a chain, and the length function that we consider is twice its height function. In addition, any increasing function over a chain is modular.
This gives us space to define latroids on the same lattice with different rank functions but the same independent sets, as in Example~\ref{ex:latroid&indep}. In order to avoid this undesirable feature, we restrict our attention to relatively complemented lattices endowed with the height function. In particular, in this setting we are able to generalize~\cite[Lemma 1.3.3]{oxleybook}. We start by observing that a relatively complemented finite lattice is atomistic.

\begin{lemma}\label{lemma:atoms}
Let $\Ll$ be a relatively complemented finite lattice. Then every $L\in\Ll$ is the join of the atoms $J$ such that $J\leq L$.
\end{lemma}

\begin{proof}
Let $J_1,\dots,J_n$ be all the atoms in $\Ll$ such that $J_i\leq L$. If  $L>J_1\vee\dots\vee J_n$, then there exists $J$ such that $J\vee(J_1\vee\dots\vee J_n)=L$ and $J\wedge(J_1\vee\dots\vee J_n)=0_{\Ll}$. Let $J_{n+1}\leq J\leq L$ be an atom. Since $J\wedge(J_1\vee\dots\vee J_n)=0_{\Ll}$, then $J_{n+1}\notin\{J_1,\dots,J_n\}$, a contradiction.
\end{proof}

\begin{lemma}\label{lemma:vee}
Let $\Ll$ be a finite complemented modular lattice with height function $\hgt$ and let $\rho:\Ll\rightarrow\Z$ be a submodular bounded increasing function. If $L_1,L_2\in\Ll$ are such that $\rho(L_1\vee L_3)=\rho(L_1)$ for each atom $L_3\leq L_2$, then $\rho(L_1\vee\ L_2)=\rho(L_1)$. 
\end{lemma}

\begin{proof}
Let $J_1,\dots,J_n$ be all the atoms in $\Ll$ such that $J_i\leq L_2$. By Lemma~\ref{lemma:atoms}, $L_2=J_1\vee\dots\vee J_n$. 
We prove the statement by induction on $n$. If $n=1$, then $L_2$ is an atom and the statement holds. For $n>1$, we have
\begin{equation*}
    \begin{split}
        2\rho(L_1)&=\rho(L_1\vee(J_1\vee\dots\vee J_{n-1}))+\rho(L_1\vee J_n)\\
        &\geq \rho((L_1\vee(J_1\vee\dots\vee J_{n-1}))\vee(L_1\vee J_n))+\rho((L_1\vee(J_1\vee\dots\vee J_{n-1}))\wedge (L_1\vee J_n))\\
        &\geq\rho(L_1\vee J_1\vee\dots\vee J_{n})+\rho(L_1)\geq 2\rho(L_1).
    \end{split}
\end{equation*}
It follows that $\rho(L_1\vee L_2)=\rho(L_1\vee J_1\vee\dots\vee J_{n})=\rho(L_1).$
\end{proof}

Both lemmas clearly fail for arbitrary lattices, e.g., over the lattice of ideals of a finite chain ring which is not a field.

\begin{example}
Let $\mathcal{L}$ be the lattice of ideals of $\Z_8$. Then $\mathcal{L}=\{0,(4),(2),\mathbb{Z}_8\}$ is a chain and the only atom is $(4)$. In particular, $(2)$ and $\mathbb{Z}_8$ are not joins of atoms. 

Consider now the $\mathbb{Z}$-latroid $(\lambda,\lVert\cdot\rVert,\Ll)$ of Example~\ref{ex:latroid&indep}. The only atom $L_3\leq\mathbb{Z}_8$ is $L_3=(4)$ and $\lambda((2)\vee (4))=\lambda((2))$, however, $\lambda((2)\vee\mathbb{Z}_8)=\lambda(\mathbb{Z}_8)=3\neq 2=\lambda((2))$.
\end{example}

The next proposition characterizes integer-valued functions that are the rank function of a $\Z$-latroid. It is a consequence of Lemma~\ref{lemma:vee}.

\begin{proposition}\label{proposition:newrankaxioms}
    Let $\Ll$ be a finite complemented modular lattice with height function $\hgt$. A function $\rho:\Ll\rightarrow\Z$ is a rank function of a $\Z$-latroid $(\rho,\hgt,\Ll)$ if and only if it satisfies the following properties:
    \begin{itemize}
        \item[R1'.] $\rho(0_{\Ll})=0$,
        \item[R2'.] $\rho(L)\leq \rho(L\vee J)\leq \rho(L)+1$, for every atom $J\in\Ll$ and every $L\in\Ll$,
        \item[R3'.] if $\rho(L)=\rho(L\vee J_1)=\rho(L\vee J_2)$ for some atoms $J_1,J_2\in\Ll$, then $\rho(L)=\rho(L\vee J_1\vee J_2)$.
    \end{itemize}
\end{proposition}
\begin{proof}
If $\rho$ is the rank function of a $\Z$-latroid, then clearly $\rho$ satisfies R1' and R2', and by Lemma~\ref{lemma:vee} it also satisfies R3'. 

Now, let $\rho:\Ll\rightarrow\Z$ be a function that satisfies R1', R2', and R3'. For $L,M\in\Ll$, $L\leq M$, it is easy to show that R2' implies L4 by induction on $\hgt(M)-\hgt(L)$. Therefore, it suffices to prove that $\rho$ is submodular. Let $L,M\in\Ll$. We proceed by induction on $\hgt(L)-\hgt(L\wedge M)$. If $\hgt(L)-\hgt(L\wedge M)=0$, then $L\leq M$ and modularity holds. Otherwise, let $J$ be an atom such that $L=L_1\vee J$ and $L\wedge M=L_1\wedge M$. By the induction hypothesis, we have
\begin{equation*}
\begin{split}
\rho(L\vee M)+\rho(L\wedge M)&=\rho((L_1\vee J)\vee M)+\rho(L_1\wedge M)\\
 &\leq \rho((L_1\vee J)\vee M)-\rho(L_1\vee M)+\rho(L_1)+\rho(M).
\end{split}
\end{equation*}
Hence, it suffices to prove that
$$\rho((L_1\vee J)\vee M)-\rho(L_1\vee M)\leq \rho(L_1\vee J)-\rho(L_1).$$
We proceed by induction on $\hgt(M)$. The case where $\hgt(M)=0$ is trivial. So suppose that $\hgt(M)>0$ and let $\bar J$ be an atom such that $M=\bar M\vee \bar J$ and $\hgt(\bar M)<\hgt(M)$. If $\rho(L_1\vee \bar M\vee \bar J)>\rho(L_1\vee \bar M)$, since $\rho((L_1\vee J)\vee \bar M\vee \bar J)\leq \rho((L_1\vee J)\vee \bar M)+1$, we have 
$\rho((L_1\vee J)\vee \bar M\vee \bar J)-\rho(L_1\vee \bar M\vee \bar J)\leq \rho((L_1\vee J)\vee \bar M)-\rho(L_1\vee \bar M)$. If instead
$\rho((L_1\vee J)\vee \bar M)>\rho(L_1\vee \bar M)$, then $\rho((L_1\vee J)\vee \bar M\vee \bar J)-\rho(L_1\vee \bar M\vee \bar J)\leq 1=\rho((L_1\vee J)\vee \bar M)-\rho(L_1\vee \bar M)$. In both cases, we obtain
    \begin{equation*}
            \rho((L_1\vee J)\vee \bar M\vee \bar J)-\rho(L_1\vee \bar M\vee \bar J)\leq \rho((L_1\vee J)\vee \bar M)-\rho(L_1\vee \bar M)\leq\rho(L_1\vee J)-\rho(L_1).
    \end{equation*}
    If instead $\rho(L_1\vee \bar M\vee \bar J)=\rho(L_1\vee \bar M)$ and $\rho((L_1\vee J)\vee \bar M)=\rho(L_1\vee \bar M)$, by R3' we obtain $\rho((L_1\vee J)\vee \bar M\vee \bar J)=\rho(L_1\vee \bar M)$, and this concludes the proof.
\end{proof}

Another important consequence of Lemma~\ref{lemma:vee} is that, for a finite complemented modular lattice with the height as length function, the set of independent elements of a $\Z$-latroid determines the rank function of the latroid, and hence the latroid itself.

\begin{proposition}\label{corollary:vee}
    Let $\Ll$ be a finite complemented modular lattice with height function $\hgt$. Let $\rho:\Ll\rightarrow\Z$ be a submodular bounded increasing function and let $L\in\Ll$. If $I$ is a maximal independent element in $[0,L]$, then $\rho(L)=\rho(I)$. In other words, the set of independent elements of a $\Z$-latroid $(\rho,\hgt,\Ll)$ uniquely determines $\rho$.
\end{proposition}

\begin{proof}
Since $I$ is maximal, we have $\rho(I\vee J)=\rho(I)$ for every atom in $L$. We conclude by Lemma~\ref{lemma:vee}.
\end{proof}

Lemma~\ref{lemma:vee} and Proposition~\ref{corollary:vee} also imply that, in the case of complemented lattices, the independent elements of a $\Z$-latroid satisfy a latroid version of the independence augmentation property.

\begin{proposition}\label{proposition:independent}
    Let $\Ll$ be a finite complemented modular lattice with height function $\hgt$. Consider a $\Z$-latroid $(\rho,\hgt,\Ll)$. The set of independent elements $\mathcal{I}$ of $\Ll$ satisfies the following properties:
    \begin{enumerate}
        \item[I1.] $0_{\Ll}\in\mathcal{I}$,
        \item[I2.] if $I_1\in\mathcal{I}$ and $I_2<I_1$, then $I_2\in\mathcal{I}$,
        \item[I3.] if $I_1,I_2\in\mathcal{I}$ and $\hgt(I_2)<\hgt(I_1)$, then there is an atom $J\leq I_1$ such that $J\nleq I_2$ and $I_2\vee J\in\mathcal{I}$. 
        \item[I4.] for any $L_1,L_2\in\Ll$ and $I_1,I_2\in\mathcal{I}$ maximal such that $I_1\leq L_1$ and $I_2\leq L_2$, there exists a maximal independent element $I_3\leq L_1\vee L_2$ contained in $I_1\vee I_2$. 
    \end{enumerate}
    Conversely, let $\mathcal{I}$ be a subset of $\Ll$ that satisfies properties I1, I2, and I4. Then there exists a unique function $\rho$ such that $(\rho,\hgt,\Ll)$ is a $\Z$-latroid whose set of independent elements is $\mathcal{I}$. Moreover, for any $L\in\Ll$, one has $\rho(L)=\hgt(I)$ for $I$ maximal among the elements of $\mathcal{I}$ which are dominated by~$L$.
\end{proposition}

\begin{proof}
Since $(\rho,\hgt,\Ll)$ is a $\Z$-latroid, we have $\rho(0_{\Ll})=0=\hgt(0_{\Ll})$, and hence $0_{\Ll}\in\mathcal{I}$. Property I2 is satisfied, since $\rho$ is bounded increasing. Now consider $I_1,I_2\in\mathcal{I}$ with $\hgt(I_2)<\hgt(I_1)$. By Lemma~\ref{lemma:atoms}, we can write $I_1=J_1\vee\dots\vee J_n$ with $\hgt(J_i)=1$ for all $i\in[n]$. Assume by contradiction that for every $J_i\nleq I_2$ we have $I_2\vee J_i\notin \mathcal{I}$. Then, for all $i\in[n]$ we have $\rho(J_i\vee I_2)<\hgt(I_2)+1$, and so $\rho(J_i\vee I_2)=\rho(I_2)$. By Lemma~\ref{lemma:vee} we obtain $\rho(I_1\vee I_2)=\rho(I_2)$, and this implies
$$\hgt(I_1)\leq \rho(I_1\vee I_2)=\rho(I_2)=\hgt(I_2),$$
that is a contradiction. This establishes Property I3. As for Property I4, by the submodularity of $\rho$ we have
\begin{equation*}
        \rho((I_1\vee I_2)\vee L_1)\leq\rho(I_1\vee I_2)+\rho(L_1)-\rho((I_1\vee I_2)\wedge L_1)\leq \rho(I_1\vee I_2)+\rho(L_1)-\rho(I_1)=\rho(I_1\vee I_2),
\end{equation*}
where the last equality follows from Proposition~\ref{corollary:vee}. Moreover, since $L_1\vee L_2\geq I_1\vee I_2$, then
\begin{equation*}
\rho(L_1\vee L_2) \geq \rho(I_1\vee I_2)
\end{equation*}
by L4.
By Proposition~\ref{corollary:vee}, a maximal independent element $I_3\in[0,I_1\vee I_2]$ has $\rho(I_3)=\rho(I_1\vee I_2)=\rho(L_1\vee L_2)$, therefore $I_3$ is also a maximal independent element in $[0,L_1\vee L_2]$.

Conversely, let $\Ll$ be a finite complemented modular lattice with height function $\hgt$ and let $\mathcal{I}$ be a subset of $\Ll$ that satisfies properties I1, I2, and I4. We claim that there exists a submodular and bounded increasing function $\rho$ such that $\rho(I)=\hgt(I)$ for every $I\in\mathcal{I}$ and that such a function is unique. Let $L\in\Ll$ and let $I$ be maximal among the elements of $\mathcal{I}$ which are dominated by $L$. By Proposition~\ref{corollary:vee} it must be $\rho(L)=\rho(I)$. This shows that the value of $\rho$ on each element of $\Ll$ is determined. It is easy to check that $\rho$ is bounded increasing, so we just need to prove that $\rho$ is submodular. Consider $L_1,L_2\in\Ll$. Let $I_3$ be a maximal independent element in $[0,L_1\wedge L_2]$ and let $I_1,I_2$  be maximal independent elements in $[I_3,L_1]$ and in $[I_3,L_2]$, respectively. We have $\rho(L_1\wedge L_2)=\rho(I_3)\leq\rho(I_1\wedge I_2)$, so equality holds. By Property I4, $\rho(L_1\vee L_2)=\rho(I_1\vee I_2)$. Therefore, using the fact that $\rho$ is bounded increasing and submodular on $\mathcal{I}$, we obtain
\begin{equation*}
        \rho(L_1\vee L_2)+\rho(L_1\wedge L_2)\leq\rho(I_1\vee I_2)+\rho(I_1\wedge I_2)\leq
        \rho(I_1)+\rho(I_2)=\rho(L_1)+\rho(L_2),
\end{equation*}
which concludes the proof.
\end{proof}

Notice that Property I3, which corresponds to the independence augmentation property of matroids, is not used in the proof of the previous proposition. Indeed, Property I3 is implied by the other properties. This is consistent with what happens in the case of $q$-matroids. We refer to~\cite{byrne2022constructions, gluesing2022independent, jurrius2009qmatroid} for the independence axioms of $q$-matroids and to~\cite{ceria2024alternatives} for a definition containing only three axioms. We chose to include Property I3 because of the next proposition, that applies to matroids among others. We start with a preliminary lemma.

\begin{lemma}\label{lemma:distr}
    Let $\Ll$ be a distributive lattice, and let $J,L_1,L_2\in\Ll$. If $J$ is an atom and $J\leq L_1\vee L_2$, then $J\leq L_1$ or $J\leq L_2$.
\end{lemma}

\begin{proof}
Since $\Ll$ is a distributive lattice,
    \begin{equation*}
        J=J\wedge(L_1\vee L_2)= (J\wedge L_1)\vee (J\wedge L_2),
    \end{equation*}
and so $J\leq L_1$ or $J\leq L_2$. 
\end{proof}

The statement of the lemma is false in general for a lattice which is not distributive.

\begin{example}
Let $\Ll$ be the lattice of linear subspaces of $\F_2^2$. Then $J=\langle(1,1)\rangle$ is an atom and $J\leq\langle(1,0)\rangle\vee\langle(0,1)\rangle$, but $J\not\leq\langle(1,0)\rangle$ and $J\not\leq\langle(0,1)\rangle$.
\end{example}

Using Lemma~\ref{lemma:distr}, we are now ready to prove that, over a finite complemented distributive lattice, property I4 is not necessary, as it is implied by the other properties of independent sets. This applies, for example, to the set of independent sets of a matroid.

\begin{proposition}\label{prop:distributive}
    Let $\Ll$ be a finite complemented distributive lattice with height function $\hgt$. Consider a $\Z$-latroid $(\rho,\hgt,\Ll)$. Then, Properties I1, I2, and I3 imply Property I4.
\end{proposition}

\begin{proof}
Consider $L_1,L_2\in\Ll$ and let $I_1,I_2\in\mathcal{I}$ be maximal such that $I_1\leq L_1$ and $I_2\leq L_2$. Let $I_3=I_1\vee J_1\vee\dots\vee J_n$ be a maximal independent element in $[I_1, L_1\vee L_2]$, where $J_1,\dots,J_n$ are all the atoms in $[0_{\Ll}, I_3]$ such that $J_i\nleq I_1$. Since $I_1$ is maximal independent in $[0_{\Ll},L_1]$, $J_i\nleq L_1$ for all $i\in[n]$, hence $I_3\wedge L_1=I_1$ by distributivity.
   
By repeatedly applying Property I3 to $I_3$ and $I_2$, we find a maximal independent element $I$ in $[I_2, L_1\vee L_2]$ with the property that $I\leq I_2\vee I_3$. By distributivity $I=(I\wedge L_1)\vee(I\wedge L_2)$, moreover $I_2\leq I\wedge L_2$ implies $I\wedge L_2=I_2$, since $I_2$ is maximal independent in $[0_{\Ll},L_2]$. Moreover, $I=I\wedge(I_2\vee I_3)=(I\wedge I_2)\vee(I\wedge I_3)=I_2\vee(I\wedge I_3)$ and $I\wedge I_3=(I\wedge I_3)\wedge(L_1\vee L_2)=(I\wedge I_3\wedge L_1)\vee(I\wedge I_3\wedge L_2)\leq (I_3\wedge L_1)\vee (I\wedge L_2)=I_1\vee I_2$. This shows that $I\leq I_1\vee I_2$.
\end{proof}

Proposition~\ref{proposition:independent} implies that a $\Z$-latroid over a finite complemented modular lattice is fully described by its set of independent elements. Similarly to what happens for matroids and $q$-matroids~\cite{byrne2022constructions,jurrius2009qmatroid,oxleybook}, latroids can also be described using bases and circuits.

\begin{proposition}\label{proposition:bases}
Let $\Ll$ be a finite complemented modular lattice with height function $\hgt$. A subset $\mathcal{B}\subseteq\Ll$ is the set of bases of a $\Z$-latroid $(\rho,\hgt,\Ll)$ if and only if 
\begin{enumerate}
    \item[B1.] $\mathcal{B}\neq\emptyset$,
    \item[B2.] if $B_1=J_1\vee\dots\vee J_n$ and $B_2=H_1\vee\dots\vee H_m$ where $J_1,\dots,J_n,H_1,\dots,H_m$ are atoms, and $J_i\nleq B_2$, then there exists an index $s\in[m]$ such that $H_s\nleq B_1$ and $J_1\vee\dots J_{i-1}\vee J_{i+1}\vee\dots\vee J_n\vee H_s\in\mathcal{B}$.
    \item[B3.] for any $L_1,L_2\in\Ll$, let $B_1,B_2\in\mathcal{B}$ be such that  $B_1\wedge L_1$ and $B_2\wedge L_2$ are maximal. Then there exists $B_3\in\mathcal{B}$ such that $B_3\wedge(L_1\vee L_2)$ is maximal and $B_3\wedge(L_1\vee L_2)\leq(B_1\wedge L_1)\vee(B_2\wedge L_2)$. 
\end{enumerate}
If the equivalent conditions hold, then for any $L\in\Ll$ we have $\rho(L)=\hgt(L\wedge B)$ for any $B\in\mathcal{B}$ such that $L\wedge B$ is maximal.  
\end{proposition}

\begin{proof}
Let $\mathcal{B}$ be a subset of $\Ll$ satisfying properties B1, B2, and B3. Consider the set $$\mathcal{I}=\{I\in\Ll:\text{there exists }B\in\mathcal{B}\text{ such that }I\leq B\}.$$
Clearly, $\mathcal{I}$ satisfies I1 and I2. In order to prove that it also satisfies I4, we notice that an element $I\in\mathcal{I}$ is maximal in $L$ if and only if there exists $B\in\mathcal{B}$ such that it has maximal intersection with $L$ among the elements in $\mathcal{B}$ and $I=L\wedge B$. Property I4 now follows directly from B3.
Moreover, notice that $\mathcal{B}$ is by definition the set of maximal elements of $\mathcal{I}$.
    
Conversely, let $\mathcal{I}$ be a subset of $\Ll$ that satisfies properties I1, I2, I3, and I4. Consider the set $$\mathcal{B}=\{B\in\mathcal{I}:B\text{ is maximal with respect to the order in the lattice}\}.$$
It is easy to check that I1 implies B1, I3 implies B2, and I4 implies B3. 
We conclude by Proposition~\ref{proposition:independent}.
\end{proof}

Notice that all bases have the same height by I3, so in B2 we have $n=m$. Similarly to the properties of independent elements, here we also have that one of the axioms is redundant. Indeed, the proof of the proposition shows that B3 is equivalent to I4 and that I3 implies B2. It follows that B3 implies B2. However, if the lattice is distributive, then the two properties are equivalent.

\begin{corollary}
Let $\Ll$ be a finite complemented distributive lattice with height function $\hgt$. Consider a $\Z$-latroid $(\rho,\hgt,\Ll)$. Then B2 implies B3.
\end{corollary}

\begin{proof}
Assume without loss of generality that $\mathcal{B}\neq\emptyset$.
It suffices to show that B2 implies I1, I2, and I3 for $\mathcal{I}=\{I\in\Ll:\text{there exists }B\in\mathcal{B}\text{ such that }I\leq B\}$. In fact, if this is the case, then they also imply I4 by Proposition~\ref{prop:distributive}. We conclude as in the proof of Proposition~\ref{proposition:bases}. 
\end{proof}

The next proposition concerns the properties of the circuits of a latroid.

\begin{proposition}\label{proposition:circuits}
    Let $\Ll$ be a finite complemented modular lattice with height function $\hgt$. A subset $\mathcal{C}\subseteq\Ll$ is the set of circuits of a $\Z$-latroid $(\rho,\hgt,\Ll)$ if and only if 
    \begin{enumerate}
        \item[C1.] $0_{\Ll}\notin\mathcal{C}$,
        \item[C2.] if $C_1,C_2\in\Cc$ are such that $C_1\leq C_2$, then $C_1=C_2$,
        \item[C3.] if $C_1,C_2\in\Cc$  are distinct and $L\leq C_1\vee C_2$ is such that $\hgt(L)=\hgt(C_1\vee C_2)-1$, then there exists $C_3\in\Cc$ such that $C_3\leq L$.
    \end{enumerate}
\end{proposition}

We start by proving some preliminary results.

\begin{lemma}\label{lemma:strongC3}
    Let $\Ll$ be a finite complemented modular lattice with height function $\hgt$ and let $\Cc$ be a subset of $\Ll$ satisfying C1, C2, and C3. Then, for every $C_1,C_2\in\Cc$ and $L\leq C_1\vee C_2$ such that $\hgt(L)=\hgt(C_1\vee C_2)-1$ and $C_2\nleq L$, we have
    \begin{equation*}
        C_1\vee C_2=\bigvee\{C\in\Cc:C\leq L\}\vee C_2.
    \end{equation*}
\end{lemma}

\begin{proof}
If $C_1=C_2$, then $L<C_1$ does not contain any element of $\Cc$ by C2 and $C_1\vee C_2=C_2$ holds.
Hence, let $C_1,C_2$ be a pair of distinct elements in $\Cc$ for which the statement fails and such that $C_1\vee C_2$ is minimal with such property. For such a pair, it holds that $C_1,C_2\not\leq L$.
By C3, there exists $C_3\in\Cc$ such that $C_3\leq L$ and $C_3\vee C_2<C_1\vee C_2$. Let $C_1\leq \bar L\leq C_1\vee C_2$ be such that $\hgt(\bar L)=\hgt(C_1\vee C_2)-1$, then 
\begin{equation*}
    \hgt(\bar L\wedge(C_2\vee C_3))=\hgt(\bar L)+\hgt(C_2\vee C_3)-\hgt(\bar L\vee C_2\vee C_3)=\hgt(C_2\vee C_3)-1,
\end{equation*}
since $\bar L\vee C_2\vee C_3=C_1\vee C_2$.
Applying C3 again, we find $C_4\leq \bar L\wedge(C_2\vee C_3)$. We notice that 
    \begin{equation*}
        \hgt(L\wedge(C_1\vee C_4))=\hgt(L)+\hgt(C_1\vee C_4)-\hgt(L\vee (C_1\vee C_4))= \hgt(C_1\vee C_4)-1.
    \end{equation*}
The last equality follows from observing that $L\leq L\vee(C_1\vee C_4)=C_1\vee C_2$, and the inequality is strict since $C_1\not\leq L$. Moreover, notice that $C_1\neq C_4$, since $C_4\leq C_2\vee C_3$, but $C_1\not\leq C_2\vee C_3$.
Since $C_1\vee C_4\leq \bar L<C_1\vee C_2$, by the minimality of $C_1\vee C_2$, we obtain
\begin{equation*}
C_1\vee C_4=\bigvee\{C\in\Cc:C\leq L\wedge(C_1\vee C_4)\}\vee C_1\leq \bigvee\{C\in\Cc:C\leq L\}\vee C_1.
\end{equation*}
Since $C_4\leq C_2\vee C_3$ and $C_3\leq L$, we conclude that $C_1\vee C_2=\bigvee\{C\in\Cc:C\leq L\}\vee C_2$.
\end{proof}

\begin{lemma}\label{lemma:2strongC3}
Let $\Ll$ be a finite complemented modular lattice with height function $\hgt$ and let $\Cc$ be a subset of $\Ll$ satisfying C1, C2, and C3. Let $L_1,L_2\in\Ll$, $L_1\geq L_2$ be such that $\hgt(L_2)=\hgt(L_1)-1$. If there exist $\bar C\in\Cc$ such that $\bar C\leq L_1$ and $\bar C\nleq L_2$, then 
\begin{equation*}
\bigvee\{C\in\Cc:C\leq L_1\}=\bigvee\{C\in\Cc:C\leq L_2\}\vee \bar C.
\end{equation*}
\end{lemma}

\begin{proof}
Since $L_2\leq L_1$ and $\hgt(L_2)=\hgt(L_1)-1$, for every $D \in\Cc$ with $D\leq L_1$ we have $L_2\leq D\vee\bar C\vee L_2\leq L_1$, hence 
$D\vee\bar C\vee L_2=L_1$, since $\bar C\not\leq L_2$. Therefore
$$\hgt((D\vee\bar C)\wedge L_2)=\hgt(D\vee\bar C)+\hgt(L_2)-\hgt(L_1)=\hgt(D\vee\bar C)-1.$$
Let $L_3=(D\vee\bar C)\wedge L_2$. By Lemma~\ref{lemma:strongC3}
\begin{equation*}
D\vee\bar C =\bigvee\{C\in\Cc:C\leq L_3\}\vee\bar C \leq \bigvee\{C\in\Cc:C\leq L_2\}\vee\bar C.
\end{equation*}
We conclude by taking the join over $D\in \Cc$, $D\leq L_1$ on both sides.
\end{proof}

We introduce the concept of chain of circuits. These are ordered sequences of circuits, so that the joins of the first $i$, for every $i$, form a strictly ascending chain of lattice elements. This concept is used in the proof of Proposition \ref{proposition:circuits}.

\begin{definition}
    Let $\Ll$ be a finite complemented modular lattice with height function $\hgt$ and let $\Cc$ be a subset of $\Ll$ satisfying C1, C2, and C3. A {\bf chain} in $\Cc$ is a sequence of elements $C_1,\dots, C_n\in\Cc$ such that $C_1\vee\dots\vee C_i<C_1\vee \dots\vee C_{i+1}$ for $1\leq i\leq n-1$. A chain $C_1,\dots, C_n\in\Cc$ is {\bf dominated} by $L$ if $C_i\leq L$ for all $i\in[n]$. A chain is {\bf maximal} if it cannot be refined, i.e., it is not a proper subsequence of another chain of circuits.
\end{definition}

\begin{lemma}\label{lemma:lengthmaxchain}
     Let $\Ll$ be a finite complemented modular lattice with height function $\hgt$ and let $\Cc$ be a subset of $\Ll$ satisfying C1, C2, and C3. Let $L\in\Ll$. Any maximal chain in $\Cc$ dominated by $L$ has the same length.
\end{lemma}

\begin{proof}
We proceed by induction on the height of $L$. If $\hgt(L)=0$, then $L$ contains no element of $\Cc$ and the length of any chain dominated by $L$ is $0$. Assume therefore that $\hgt(L)>0$ and that the statement holds for every $L'<L$. Let $C_1,\dots,C_m\in\Cc$ and $D_1,\dots,D_n\in\Cc$ be two maximal chains dominated by $L$ with $m\leq n$. Let $C_1\vee\dots\vee C_{m-1}\leq\bar L<L$ be such that $\hgt(\bar L)=\hgt(L)-1$ and $C_1\vee\dots\vee C_m\nleq \bar L$. Such an $\bar L$ exists since $C_1\vee\dots\vee C_{m-1}<C_1\vee\dots\vee C_m\leq L$ and $\mathcal{L}$ is relatively complemented. If $D_i\leq \bar L$ for every $i\in[n]$, then $D_1,\dots,D_n$ is a maximal chain dominated by $\bar L$. Else, let $1\leq i\leq n$ be the smallest index for which $D_i\nleq\bar L$. For every $i<j\leq n$, $\bar L\vee D_i\vee D_j=L$, and hence $\hgt(\bar L\wedge(D_i\vee D_j))=\hgt(D_i\vee D_j)-1$.
By Lemma~\ref{lemma:strongC3} there exists $\bar D_j\leq \bar L\wedge(D_j\vee D_i)$ such that $\bar D_j\nleq D_1\vee\dots\vee D_{j-1}$. In fact, if this were not the case, we would obtain
\begin{equation*}
    D_j<D_j\vee D_i=\bigvee\{C\in\Cc:C\leq \bar L\wedge(D_j\vee D_i)\}\vee D_i\leq D_1\vee\dots\vee D_{j-1},
\end{equation*}
that contradicts our assumptions. Consider the sequence $D_1,\dots, D_{i-1},\bar D_{i+1},\dots,\bar D_n$. This is a chain of circuits in $\bar L$. Indeed, we have
\begin{equation*}
D_1\vee\dots\vee D_{i-1}\vee \bar D_{i+1}\vee\ldots\vee\bar D_j\leq D_1\vee\dots\vee D_j,
\end{equation*}
while $\bar D_{j+1}\nleq D_1\vee\dots\vee D_j$. Since $C_1,\dots, C_m$ is a maximal chain in $L$, $C_1,\dots, C_{m-1}$ has to be maximal in $\bar L$. Indeed, suppose that there exists a circuit $\bar C$ such that $C_1\dots,C_{i-1},\bar C,C_i\dots, C_{m-1}$ is a chain in $\Cc$ dominated by $\bar L$, then $C_1\dots,C_{i-1},\bar C,C_i\dots, C_{m}$ is a chain in $\Cc$ dominated by $L$, but this contradicts the maximality of $C_1,\dots,C_m$. Therefore, we have $n-1\leq m-1$, i.e., $n\leq m$.
\end{proof}

We are now ready to prove Proposition~\ref{proposition:circuits}.

\begin{proof}[Proof of Proposition~\ref{proposition:circuits}]
Let $\Cc$ be the set of circuits of a $\Z$-latroid $(\rho,\hgt,\Ll)$. By definition $\rho(0_{\Ll})=0$, and so $0_{\Ll}\notin \Cc$. Moreover, if $C_1< C_2$ and $C_1$ is a circuit, we have $C_2\notin\Cc$, which proves C2. Finally, let $C_1$ and $C_2$ be distinct circuits. By the submodularity of $\rho$, we obtain
\begin{equation*}
\rho(C_1\vee C_2)\leq \rho(C_1)+\rho(C_2)-\rho(C_1\wedge C_2)\leq \hgt(C_1)-1+\hgt(C_2)-1-\hgt(C_1\wedge C_2)=\hgt(C_1\vee C_2)-2.
\end{equation*}
Since $\rho$ is bounded increasing, then $$\rho(L)\leq\rho(C_1\vee C_2)\leq\hgt(C_1\vee C_2)-2=\hgt(L)-1$$ for every $L\leq C_1\vee C_2$ with $\hgt(L)=\hgt(C_1\vee C_2)-1$. Hence $L$ is dependent, so there exists a minimal dependent element $C_3$ such that $C_3\leq L$. This implies C3.

Let $\Cc$ be a subset of $\Ll$ satisfying C1, C2, and C3. Consider the function $\kappa:\Ll\rightarrow \Z$ that associates to an element $L\in\Ll$ the length of a maximal chain of elements in $\Cc$ dominated by $L$. This function is well defined by Lemma~\ref{lemma:lengthmaxchain}. 
We claim that for $L_1,L_2\in\Ll$ we have
\begin{equation}\label{equation:kappaissupmodular}
\kappa(L_1)+\kappa(L_2)\leq \kappa(L_1\vee L_2)+\kappa(L_1\wedge L_2).
\end{equation}
We prove the claim by induction on $\hgt(L_1)-\hgt(L_1\wedge L_2)$. If $\hgt(L_1)-\hgt(L_1\wedge L_2)=0$, then $L_1\leq L_2$ and~\eqref{equation:kappaissupmodular} is an equality. If $\hgt(L_1)-\hgt(L_1\wedge L_2)=\hgt(L_1\vee L_2)-\hgt(L_2)>0$, then there exists $L<L_1\vee L_2$ with $\hgt(L)=\hgt(L_1\vee L_2)-1$ and such that $L_2\leq L$. If all the circuits of $L_1$ are contained in $L_1\wedge L$, then we have 
\begin{equation*}
\kappa(L_1)+\kappa(L_2)=\kappa(L_1\wedge L)+\kappa(L_2)\leq \kappa((L_1\wedge L)\vee L_2)+\kappa((L_1\wedge L)\wedge L_2)\leq \kappa(L_1\vee L_2)+\kappa(L_1\wedge L_2),
\end{equation*}
where the second inequality follows from the induction hypothesis, since $L_1\nleq L$ implies $\hgt(L_1\wedge L)-\hgt((L_1\wedge L)\wedge L_2)<\hgt(L_1)-\hgt(L_1\wedge L_2)$. Assume now that there exists a circuit $\bar C\leq L_1$ such that $\bar C\nleq L$. Since $\hgt(L\wedge L_1)=\hgt(L_1)+\hgt(L_1\vee L_2)-1-\hgt(L_1\vee L_2)=\hgt(L_1)-1$, by Lemma~\ref{lemma:2strongC3} we have that 
\begin{equation*}
\bigvee\{C\in\Cc:C\leq L_1\}=\bigvee\{C\in\Cc:C\leq L_1\wedge L\}\vee \bar C.
\end{equation*}
This implies $\kappa(L_1)=\kappa(L_1\wedge L)+1$. Similarly, applying Lemma~\ref{lemma:2strongC3} to $L_1\vee L_2$ and $L$, we obtain $\kappa(L_1\vee L_2)=\kappa(L)+1$. Moreover, by modularity we have $(L_1\wedge L)\vee L_2=L$, hence
\begin{equation*}
\kappa(L_1)+\kappa(L_2)=\kappa(L_1\wedge L)+\kappa(L_2)+1\leq \kappa(L)+1+\kappa(L_1\wedge L_2)=\kappa(L_1\vee L_2)+\kappa(L_1\wedge L_2),
\end{equation*}
where the first inequality follows from the induction hypothesis, since $\hgt(L_1\wedge L)-\hgt((L_1\wedge L)\wedge L_2)<\hgt(L_1)-\hgt(L_1\wedge L_2)$. 

The function $\rho=\hgt-\kappa$ is submodular, since $\kappa$ satisfies Equation~\ref{equation:kappaissupmodular}. Let $L_2\leq L_1\in\Ll$, then $\kappa(L_1)\geq \kappa(L_2)$, and so $\rho(L_1)-\rho(L_2)\leq \hgt(L_1)-\hgt(L_2)$. Repeated application of Lemma~\ref{lemma:2strongC3} yields
$$\kappa(L_1)-\kappa(L_2)\leq \hgt(L_1)-\hgt(L_2),$$
or equivalently, $\rho(L_1)\geq\rho(L_2)$.
Finally, let $C$ be a circuit of the $\Z$-latroid $(\hgt-\kappa, \hgt,\Ll)$. By definition, $\kappa(C)>0$, so there exists $\bar C\in\Cc$ such that $\bar C\leq C$. Since $C$ is a circuit, for every $L<C$ we have $\kappa(L)=0$. We conclude that $C=\bar C$, therefore $C\in\Cc$. Conversely, any $C\in\Cc$ has $\kappa(C)=1$, so $\rho(C)<\hgt(C)$. Moreover, if $L\in\Ll$ and $L<C$, then by C2 there is no $\bar C\in \Cc$ such that $\bar C\leq L$, hence $\kappa(L)=0$, i.e., $\rho(L)=\hgt(L)$. 
\end{proof}

Similarly to independent sets, bases, and circuits, many standard concepts in matroid theory - including closure function, flats, and hyperplanes - can be extended to latroids. For instance, one can define the closure operator as follows. 

\begin{definition}
Let $(\rho,\lVert \cdot\rVert, \mathcal{L})$ be an $A$-latroid. The {\bf closure} operator $\mathrm{cl}$ is defined as
$$\mathrm{cl}(L)=\bigvee\{\bar L\in\Ll:\rho(L\vee\bar L)=\rho(L)\},$$
for all $L\in\mathcal{L}$. 
\end{definition}

One always has $L\leq \mathrm{cl}(L)$. Moreover, for $L_1\leq L_2$ we have $\mathrm{cl}(L_1)\leq \mathrm{cl}(L_2)$. In fact, let $L$ be such that $\rho(L\vee L_1)=\rho(L_1)$. Then
\begin{equation*}
    \begin{split}
        \rho(L_2)+\rho(L_1\vee L)&\geq \rho(L_2\vee(L_1\vee L))+\rho(L_2\wedge(L_1\vee L))\\
        &=\rho(L_2\vee L)+\rho(L_1\vee(L_2\wedge L))\geq \rho(L_2\vee L)+\rho(L_1),    
    \end{split}
\end{equation*}
where in the equality we use the fact that $\Ll$ is modular. The previous inequality is equivalent to $\rho(L_2\vee L)-\rho(L_2)\leq\rho(L_1\vee L)-\rho(L_1)=0$, so $\rho(L_2)=\rho(L_2\vee L)$. 

\begin{definition}
Let $(\rho,\lVert \cdot\rVert, \mathcal{L})$ be an $A$-latroid. 
If $L=\mathrm{cl}(L)$, we call $L$ a {\bf flat}. A {\bf hyperplane} is a flat $L$ such that $\rho(L)=\rho(1_{\Ll})-1$.    
\end{definition}

Closure function, flats, and hyperplanes - similarly to independent elements - acquire greater significance in the case of a latroid built on a complemented lattice. By carefully adapting the proofs in~\cite{byrne2022constructions}, one can find cryptomorphic definitions of a latroid based on these notions. For cyclic flats we also refer to~\cite{byrne2023cyclic}.

\begin{proposition}\label{proposition:closureaxioms}
Let $\Ll$ be a finite complemented modular lattice with height function $\hgt$. A function $\mathrm{cl}:\Ll\rightarrow\Ll$ is the closure function of a $\Z$-latroid $(\rho,\hgt,\Ll)$ if and only if 
\begin{enumerate}
    \item[CL1.] $L\leq\mathrm{cl}(L)$ for all $L\in\Ll$,
    \item[CL2.] if $L_1\leq L_2\in\Ll$, then $\mathrm{cl}(L_1)\leq\mathrm{cl}(L_2)$,
    \item[CL3.] $\mathrm{cl}(L)=\mathrm{cl}(\mathrm{cl}(L))$ for all $L\in\Ll$,
    \item[CL4.] if $L,J_1,J_2\in\Ll$ with $\hgt(J_1)=\hgt(J_2)=1$, $J_2\leq\mathrm{cl}(L\vee J_1)$, and $J_2\nleq\mathrm{cl}(L)$, then $J_1\leq\mathrm{cl}(L\vee J_2)$.
\end{enumerate}
If the equivalent conditions hold, then the rank function of the latroid with closure function $\mathrm{cl}$ is
$$\rho(L)=\min\{\hgt(M):\mathrm{cl}(M)=\mathrm{cl}(L)\text{ and }M\leq L\}.$$
\end{proposition}

\begin{proof}
We already saw in the discussion above that the closure function of a latroid satisfies CL1 and CL2 even in the case that the lattice is not complemented. By Lemma~\ref{lemma:vee} we have $\rho(\mathrm{cl}(L))=\rho(L)$. Fix $\bar L\in\Ll$. If $\rho(L\vee \bar L)=\rho(L)$, then $\bar L\leq\mathrm{cl}(L)$. If, instead, $\rho(L\vee \bar L)>\rho(L)$, then
$$\rho(\bar L\vee\mathrm{cl}(L))\geq\rho(L\vee \bar L)>\rho(L)=\rho(\mathrm{cl}(L)).$$
We conclude that $\mathrm{cl}(L)=\mathrm{cl}(\mathrm{cl}(L))$. To prove CL4, let $L,J_1,J_2\in\Ll$ be such that $\hgt(J_1)=\hgt(J_2)=1$, $J_2\leq\mathrm{cl}(L\vee J_1)$, and $J_2\not\leq\mathrm{cl}(L)$. We obtain  $\rho(L\vee J_1\vee J_2)=\rho(L\vee J_1)\leq \rho(L)+1$. Moreover, since $J_2\not\leq\mathrm{cl}(L)$, we have $\rho(L\vee J_2)=\rho(L)+1$. Combining these two equations we find $\rho(L\vee J_2)\geq \rho(L\vee J_1\vee J_2)$, and this implies that $J_1\in\mathrm{cl}(L\vee J_2)$.

Conversely, let $\mathrm{cl}:\Ll\rightarrow\Ll$ be a function satisfying the four axioms in the statement of the proposition. We define the function $\rho_{\mathrm{cl}}:\Ll\rightarrow \Z$ as
\begin{equation*}
\rho_{\mathrm{cl}}(L)=\min\{\hgt(M):\mathrm{cl}(M)=\mathrm{cl}(L)\text{ and }M\leq L\}.
\end{equation*}
We claim that $\rho_{\mathrm{cl}}$ is the rank function of a $\Z$-latroid. By Proposition~\ref{proposition:newrankaxioms} it suffices to show that $\rho_{\mathrm{cl}}$ satisfies R1', R2', and R3'. By definition, we have $\rho_{\mathrm{cl}}(0_{\Ll})=0$. 

Let $L\in\Ll$, let $J\in\Ll$ be an atom. We start by showing that $\rho_{\mathrm{cl}}(L)\leq\rho_{\mathrm{cl}}(L\vee J)$. We may suppose that $J\not\leq L$. By the definition of $\rho_{\mathrm{cl}}$, there exist $H\leq L$, $H'\leq L\vee J$ such that $\mathrm{cl}(H)=\mathrm{cl}(L)$, $\mathrm{cl}(H')=\mathrm{cl}(L\vee J)$, $\hgt(H)=\rho_{\mathrm{cl}}(L)$, and $\hgt(H')=\rho_{\mathrm{cl}}(L\vee J)$. If $H'\leq L$, then we conclude. Otherwise, let $\bar J$ and $\bar H$ be such that $H'=\bar H\vee \bar J$ with $\bar J$ an atom and $\bar H\leq L$. Either $\mathrm{cl}(\bar H)=\mathrm{cl}(L)$ leading to $\rho_{\mathrm{cl}}(L)<\rho_{\mathrm{cl}}(L\vee J)$, or there exists an atom $\tilde J\leq L$ such that $\tilde J\nleq \mathrm{cl}(\bar H)$ and $\tilde J\leq \mathrm{cl}(L)\leq\mathrm{cl}(\bar H\vee \bar J)$. Then, by CL4 we have $\bar J\leq \mathrm{cl}(\bar H\vee \tilde J)$. Hence, by CL2 and CL3 we find $\bar H\vee\bar J\leq \mathrm{cl}(L)$, which implies
$\mathrm{cl}(L\vee J)=\mathrm{cl}(L).$
It follows that $\rho_{\mathrm{cl}}(L)\leq \rho_{\mathrm{cl}}(L\vee J)$. 
Moreover, by properties CL1, CL2, and CL3, we have $\mathrm{cl}(H\vee J)\leq\mathrm{cl}(L\vee J)=\mathrm{cl}(\mathrm{cl}(H)\vee J)\leq\mathrm{cl}(\mathrm{cl}(H\vee J))=\mathrm{cl}(H\vee J)$, where the last inequality follows from $J\leq\mathrm{cl}(J)\leq\mathrm{cl}(H\vee J)$ and $\mathrm{cl}(H)\leq \mathrm{cl}(H\vee J)$.
It follows that $\rho_{\mathrm{cl}}(L\vee J)\leq \hgt(H\vee J)\leq \hgt(H)+1=\rho_{\mathrm{cl}}(L)+1$, which completes the proof of R2'.

Let $L$ be an element of $\Ll$ and $J_1,J_2\in\Ll$ be atoms such that $\rho_{\mathrm{cl}}(L)=\rho_{\mathrm{cl}}(L\vee J_1)=\rho_{\mathrm{cl}}(L\vee J_2)$. In the same way as we proved that $\rho_{\mathrm{cl}}$ is weakly increasing, we can show that $\mathrm{cl}(L)=\mathrm{cl}(L\vee J_1)=\mathrm{cl}(L\vee J_2)$ and, therefore, $\mathrm{cl}(L)=\mathrm{cl}(L\vee J_1\vee J_2)$. By CL4, this implies $\rho_{\mathrm{cl}}(L)=\rho_{\mathrm{cl}}(L\vee J_1\vee J_2)$, proving R3'.

Finally, let $(\rho, \hgt,\Ll)$ be a $\Z$-latroid. We need to verify that
\begin{equation}\label{eq:closure1}
\rho(L)=\min\{\hgt(M):\mathrm{cl}(M)=\mathrm{cl}(L)\text{ and }M\leq L\},
\end{equation}
and 
\begin{equation}\label{eq:closure2}
\mathrm{cl}(L)=\bigvee\{M\in\Ll:\rho_{\mathrm{cl}}(L\vee M)=\rho_{\mathrm{cl}}(L)\}.
\end{equation}
Equation~\eqref{eq:closure1} follows from the fact that if $M\leq L$ and $\rho(M)=\rho(L)$, then $\mathrm{cl}(M)=\mathrm{cl}(L)$, and from the fact that $\rho(L)=\hgt(M)$, where $M$ is a maximal independent element among those that are smaller than or equal to $L$. This also implies that $\rho=\rho_{\mathrm{cl}}$. 
Equation~\eqref{eq:closure2} now follows.
\end{proof}

Next, we show that a latroid is determined by its set of flats and describe the closure function of a latroid in terms of its flats.

\begin{proposition}
    Let $\Ll$ be a finite complemented modular lattice with height function $\hgt$. A~set $\mathcal{F}\subseteq \Ll$ is the set of flats of a $\Z$-latroid $(\rho,\hgt,\Ll)$ if and only if 
    \begin{itemize}
        \item[F1.] $1_{\Ll}\in\mathcal{F}$,
        \item[F2.] for every $F_1$ and $F_2\in\mathcal{F}$, $F_1\wedge F_2\in\mathcal{F}$,
        \item[F3.] for every $F_1\in\mathcal{F}$ and any atom $J\nleq F_1$, there exists a unique element $F_2\in\mathcal{F}$ such that $F_2$ covers $F_1$ and $J\leq F_2$.
    \end{itemize}
If the equivalent conditions hold, then the closure function of the latroid is $$\mathrm{cl}(L)=\min\{F\in\mathcal{F}:L\leq F\}.$$
\end{proposition}

\begin{proof}
Let $\mathcal{F}$ be the set of flats of a $\Z$-latroid $(\rho,\hgt,\Ll)$. Clearly, $1_\Ll\in\mathcal{F}$. For every $F_1,F_2\in\mathcal{F}$, by CL2 we have $$\mathrm{cl}(F_1\wedge F_2)\leq \mathrm{cl}(F_1)\wedge\mathrm{cl}(F_2)=F_1\wedge F_2,$$
hence, by CL1 we conclude that $F_1\wedge F_2\in\mathcal{F}$. Moreover, for any atom $J\nleq F_1$ we have $\mathrm{cl}(F_1\vee J)\in\mathcal{F}$. Let $F_2\in\mathcal{F}$ such that $F_1<F_2\leq\mathrm{cl}(F_1\vee J)$. Then, there exists an atom $\bar J\leq F_2\leq\mathrm{cl}(F_1\vee J)$ such that $\bar J\nleq F_1$. By CL4, we conclude that $J\leq \mathrm{cl}(F_1\vee \bar J)\leq F_2$. Therefore $F_2=\mathrm{cl}(F_1\vee J)$, i.e., $\mathrm{cl}(F_1\vee J)$ covers $F_1$.

Conversely, suppose that $\mathcal{F}\subseteq\Ll$ satisfies F1, F2 and F3. We define $f:\Ll\rightarrow \Ll$ as 
$$f(L)=\min\{F\in\mathcal{F}:L\leq F\}.$$
This function is well-defined, since $\mathcal{F}$ is closed under meet and $1_{\Ll}\in\mathcal{F}$. We claim that $f$ is a closure function. By the definition, $L\leq f(L)$ and $f(f(L))=f(L)$. This proves CL1 and CL3. Let $L_1\leq L_2$, then $f(L_2)\geq L_2\geq L_1$. This implies that $L_1\leq f(L_1)\wedge f(L_2)$ and by F2 $f(L_1)\wedge f(L_2)\in\mathcal{F}$. By the minimality of $f(L_1)$ we conclude that $f(L_1)=f(L_1)\wedge f(L_2)\leq f(L_2)$. This shows that $f$ satisfies CL2. Finally, consider two distinct atoms $J_1$ and $J_2$, such that $J_2\leq f(L\vee J_1)$ and $J_2\nleq f(L)$. If $J_1\leq f(L)$, then $J_1\leq f(L\vee J_2)$. Else, by F3 we know that both $f(L\vee J_2)$ and $f(L\vee J_1)$ cover $f(L)$. Since $L\vee J_2\leq f(L\vee J_2)\leq f(L\vee J_1)$, then $f(L\vee J_2)=f(L\vee J_1)$, in particular $J_1\leq f(L\vee J_2)$. This proves CL4.

Finally, let $(\rho, \hgt,\Ll)$ be a $\Z$-latroid. We claim that $\mathrm{cl}=f$. Notice that $\mathrm{cl}(L)\leq\mathrm{cl}(f(L))=f(L)$. Therefore, since $\mathrm{cl}(L)\in\mathcal{F}$ and by the minimality of $f(L)$, we conclude $\mathrm{cl}(L)=f(L)$. Notice moreover that, by the definition of $f$, $\mathcal{F}=\{L\in\Ll:f(L)=L\}=f(\Ll)$.
\end{proof}

Notice that F1 and F2 imply that the set of flats of a $\Z$-latroid is a lattice. We conclude this subsection by proving that the lattice of flats of a $\Z$-latroid is a geometric lattice. This extends similar results for matroids~\cite{oxleybook} and $q$-matroids~\cite{bollen2017tutte,byrne2022constructions}.

\begin{theorem}\label{thm:latticeofflats}
Let $\Ll$ be a finite complemented modular lattice with height function $\hgt$. A lattice is geometric if and only if it is the lattice of flats of a $\Z$-latroid $(\rho,\hgt,\Ll)$.
\end{theorem}

\begin{proof}
By~\cite[Theorem 1.7.5]{oxleybook} we know that a lattice is geometric if and only if it is the lattice of flats of a matroid. Since matroids are $\Z$-latroids, it suffices to show that the lattice of flats of a $\Z$-latroid is geometric. Since there are two lattices involved in this proof, to avoid ambiguities, we will denote by $\wedge_{\mathcal{F}}$ and $\vee_{\mathcal{F}}$ the operations in $\mathcal{F}$ and by $\wedge_{\Ll}$ and $\vee_{\Ll}$ the operations in $\Ll$. First of all, if $F_1,F_2\in\mathcal{F}$ and $F_2$ covers $F_1$, then there exists an atom $J\in\Ll$, $J\leq F_2$ such that $J\nleq F_1$. Since $F_1=\mathrm{cl}(F_1)$, we have $\rho(F_1\vee_\Ll J)=\rho(F_1)+1$. Moreover, $\mathrm{cl}(F_1\vee_{\Ll} J)\leq F_2$, and since $F_2$ covers $F_1$, we obtain $\mathrm{cl}(F_1\vee_{\Ll} J)= F_2$. We conclude that $\rho(F_2)=\rho(F_1)+1$. This implies that the height of a flat in the lattice $\mathcal{F}$ is equal to $\rho(F)$. 
It follows from F2 that $F_1\wedge_{\Ll} F_2=F_1\wedge_{\mathcal{F}} F_2$.
Moreover, since $\mathrm{cl}(F_1\vee_{\Ll} F_2)\in\mathcal{F}$ is the smallest element of $\mathcal{F}$ that dominates both $F_1$ and $F_2$, then $F_1\vee_{\mathcal{F}} F_2=\mathrm{cl}(F_1\vee_{\Ll} F_2)$. Therefore
\begin{equation*}
\begin{split}
\rho(F_1)+\rho(F_2)&\geq\rho(F_1\vee_{\Ll}F_2)+\rho(F_1\wedge_{\Ll}F_2)=\rho(\mathrm{cl}(F_1\vee_{\Ll}F_2))+\rho(F_1\wedge_{\Ll}F_2)\\
 &=\rho(F_1\vee_{\mathcal{F}}F_2)+\rho(F_1\wedge_{\mathcal{F}} F_2).
\end{split}
\end{equation*}
Therefore, the height function of $\mathcal{F}$ is submodular and $\mathcal{F}$ is a semimodular lattice. 

To prove that $\mathcal{F}$ is atomistic, we proceed by induction on $\rho(F)$. If $\rho(F)=1$, then $F$ is an atom in $\mathcal{F}$. Assume that every flat $H$ such that $\rho(H)<\rho(F)$ is the meet of a finite set of atoms. Let $\bar F\in\mathcal{F}$ be a flat that is covered by $F$ and let $L\in\Ll$ be such that $F=\bar F\vee_{\Ll} L$ and $\bar F\wedge_{\Ll} L=0_{\Ll}$. Since $\rho(\bar F)+1=\rho(F)\leq \rho(\bar F)+\rho(L)$, we have $\rho(L)\geq 1$. Let $\bar L\leq L$ in $\Ll$ be such that $\rho(\bar L)=1$. Then, $\bar F<\bar F\vee_{\mathcal{F}}\mathrm{cl}(\bar L)\leq F$, as $\bar L\leq F$. However, since $F$ covers $\bar F$, we obtain $F=\bar F\vee_{\mathcal{F}}\mathrm{cl}(\bar L)$. By the induction hypothesis, both $\bar F$ and $\mathrm{cl}(\mathrm{\bar L})$ are the meet of a finite set of atoms, hence so is~$F$.
\end{proof}

\begin{remark}\label{rmk:ideal_rk}
In the previous theorem we proved that the lattice of flats of a $\Z$-latroid $(\rho,\hgt,\Ll)$ is a geometric lattice. Moreover, by~\cite[Theorem 1.7.5]{oxleybook} we know that a lattice is geometric if and only if it is the lattice of flats of a matroid. This is also true if we restrict to $q$-matroids, i.e., for every $q$-matroid there exists a matroid with the same lattice of flats~\cite{jany2023projectivization}. Therefore, one can associate to a $q$-matroid the Stanley-Reisner ideal of the matroid with the same lattice of flats. In the case of the $q$-matroid associated to an $\F_{q^m}$-linear rank-metric code, the graded Betti numbers of this ideal compute the generalized weights of the code, as shown in~\cite{johnsen2023weight}.

Theorem~\ref{thm:latticeofflats} allows us to extend the construction to $\Z$-latroids. In fact, one may associate to any $\Z$-latroid the Stanley-Reisner ideal of the matroid with the same lattice of flats, whose existence is guaranteed by Theorem~\ref{thm:latticeofflats}. However, one cannot hope that the smallest graded Betti numbers of that ideal in each homological position coincide with the generalized weights of the code, in this generality. One obstruction is that the graded Betti numbers of the ideal in each homological position form a strictly increasing integer sequence, while the generalized weights are only weakly increasing in general. E.g., for $\F_{q}$-linear rank-metric codes in $\F_q^{m\times n}$, up to $m$ consecutive generalized weights may coincide.
\end{remark}

\subsection{Generalized weights of a latroid}

We conclude the section on latroids by defining the generalized weights of a latroid. 
\begin{definition}
Let $(\rho,\lVert\cdot\rVert,\Ll)$ be an $A$-latroid and $a\in A$. The {\bf $a$-generalized weight} $d_a(\Cc)$ of $(\rho,\lVert\cdot\rVert,\Ll)$ is
\begin{equation*}
    d_a(\rho,\lVert\cdot\rVert,\Ll)=\min_{L\in\Ll}\{\lVert L\rVert:\lVert L\rVert-\rho(L)\geq a\},
\end{equation*}
where set $d_a(\Cc)=0$ if the right hand side is empty.
\end{definition}

\begin{remark}
In Remark~\ref{remark:latroidtomatroid} we show that a matroid $(E,\rho)$ and a $\Z$-latroid $(\rho,\lvert\cdot\rvert,\mathcal{P}(E))$ are essentially the same mathematical object. Therefore, we can associate a latroid to a linear block code $\Cc\subseteq\F_q^n$ in the same way that we usually associate a matroid to it. We consider the lattice $\mathcal{P}([n])$ of subsets of $[n]$ and define $\rho_{\Cc}:\mathcal{P}([n])\rightarrow\mathbb{Z}$ as $\rho_{\Cc}(L)=\lvert L\rvert-\dim(\Cc(L))$, where $\Cc(L)$ is the largest subcode of $\Cc$ with Hamming support contained in $L$. It is easy to verify that $(\rho_{\Cc},\lvert\cdot\rvert,\mathcal{P}([n]))$ is a $\Z$-latroid. Moreover
    \begin{equation*}
        d_r(\rho_{\Cc},\lvert\cdot\rvert,\mathcal{P}([n]))=\min_{L\in\mathcal{P}([n])}\{\lvert L\rvert:\dim(\Cc(L))\geq r\}=\min_{\Dd\subseteq\Cc}\{\lvert \supp(\Dd)\rvert:\dim(\Dd)\geq r\}=d_r(\Cc),
    \end{equation*}
    where the equality in the middle follows from the fact that $\supp(\Dd)=\supp(\Cc(\supp(\Dd)))$ and $\dim(\Dd)\leq\dim(\Cc(\supp(\Dd)))$ for a subcode $\Dd$ of $\Cc$. Therefore, the generalized weights of the latroid associated to a linear block code are equal to the generalized weights of the code.
\end{remark}

In the next proposition, we collect some basic properties of the generalized weights of a latroid.

\begin{proposition}
    Let $(\rho_1,\lVert\cdot\rVert,\Ll)$ and $(\rho_2,\lVert\cdot\rVert,\Ll)$ be $A$-latroids such that $\rho_2(L)\leq\rho_1(L)$ for all $L\in\Ll$ and let $a\in A$ be such that $d_{a}(\rho,\lVert\cdot\rVert,\Ll)\neq 0$. Then:
    \begin{enumerate}
        \item $d_{b}(\rho,\lVert\cdot\rVert,\Ll)\leq d_{a}(\rho,\lVert\cdot\rVert,\Ll)$ if $b\leq a\in A$,
        \item $d_{a}(\rho_1,\lVert\cdot\rVert,\Ll)\leq d_{a}(\rho_2,\lVert\cdot\rVert,\Ll)$.
    \end{enumerate}
    Moreover, suppose that $A=\Z$ and that $\lVert\cdot\rVert=\hgt$ is the height function of a graded lattice. Let $a,b\in\Z$, $b<a$, and suppose that there exists $\bar L\in\Ll$ such that $\lVert \bar L\rVert -\rho(\bar L)=b$. Then $$d_{b}(\rho,\lVert\cdot\rVert,\Ll)<d_{a}(\rho,\lVert\cdot\rVert,\Ll).$$
\end{proposition}

\begin{proof}
    Items 1 and 2 follow directly from the definition of generalized weights. Let $b<a$ and suppose that there exists $\bar L\in\Ll$ such that $\hgt(\bar L)-\rho(\bar L)=b$. Let $\tilde L\in\Ll$ be such that $d_a(\rho,\hgt,\Ll)=\hgt(\tilde L)$ and $\hgt(\tilde L)-\rho(\tilde L)\geq a$. Then $\hgt(\tilde L)\geq \hgt(\bar L)$, since $\rho$ is bounded increasing. If $\hgt(\tilde L)>\hgt(\bar L)$, then the thesis follows. Assume therefore that $\hgt(\bar L)=\hgt(\tilde L)$. Let $\hat L<\tilde L$ be such that $\hgt(\tilde L)=\hgt(\hat L)+1$. Then 
    $$b\leq a-1\leq\hgt(\tilde L)-\rho(\tilde L)-1\leq \hgt(\hat L)-\rho(\hat L),$$
    since $\rho$ is bounded increasing. 
    Hence $d_b(\rho,\lVert\cdot\rVert,\Ll)\leq \lVert \hat L\rVert= \lVert \bar L\rVert-1$, which proves the thesis.
\end{proof}

\section{R-linear codes}\label{section:linear}

Let $R$ be a finite ring and let $\mathcal{M}(R^n)$ be the set of submodules of $R^n$.
In this section, we discuss how to associate a latroid to an $R$-linear code.
We denote by $\mathcal{R}^n$ the set of rectangular submodules of $R^n$, i.e.,
$$\mathcal{R}^n=\{M=I_1\times\dots\times I_n\subseteq R^n: I_i\text{ is an ideal of }R\text{ for all }i\in[n]\}.$$
Notice that $\mathcal{M}(R^n)$ and $\mathcal{R}^n$ are complete lattices with respect to the sum and the intersection. For a code $\Cc$ and a strictly increasing modular function $\lVert\cdot\rVert:\mathcal{M}(R^n)\rightarrow A$, we define $\rho_{\Cc}:\mathcal{M}(R^n)\rightarrow A$ as
\begin{equation*}
    \rho_{\Cc}(M)=\lVert M\rVert -\lVert M\cap \Cc\rVert\text{ for all }M\in\mathcal{M}(R^n).
\end{equation*}
In the next proposition, we consider the restriction of $\lVert\cdot\rVert$ and of $\rho_{\Cc}$ to an arbitrary sublattice of $\mathcal{M}(R^n)$. To simplify the notation, we do not indicate the domain of the functions whenever it is clear from the context.

\begin{proposition}\label{proposition:latroidcode}
    Let $\Ll$ be a sublattice of $(\mathcal{M}(R^n),\subseteq)$, let $\Cc\in\mathcal{M}(R^n)$ be a code.
    The triple $(\rho_{\Cc},\lVert\cdot\rVert,\Ll)$ is an $A$-latroid. In particular, the triple $(\rho_{\Cc},\lVert\cdot\rVert,\mathcal{R}^n)$ is an $A$-latroid.
\end{proposition}

\begin{proof}
Let $M_1\subseteq M_2\in\Ll$. Since $M_1\cap\Cc\leq M_2\cap\Cc$ 
$$\rho_{\Cc}(M_2)-\rho_{\Cc}(M_1)\leq \lVert M_2\rVert-\lVert M_1\rVert.$$ Moreover, we have
    \begin{equation*}
    \begin{split}
        \rho_{\Cc}(M_2)-\rho_{\Cc}(M_1)&=\lVert M_2\rVert-\lVert M_1\rVert-\lVert M_2\cap\Cc\rVert+\lVert M_1\cap \Cc\rVert=\\
        &=\lVert M_2\rVert-\lVert M_1\rVert+(\lVert M_1\rVert+\lVert \Cc\rVert-\lVert M_1+\Cc\rVert)-(\lVert M_2\rVert+\lVert \Cc\rVert-\lVert M_2+\Cc\rVert)=\\
        &=\lVert M_2+\Cc\rVert-\lVert M_1+\Cc\rVert\geq0,
    \end{split}
    \end{equation*}
    hence $\rho_{\Cc}$ is bounded increasing. We claim that $\rho_{\Cc}$ is a submodular function. Let $L_1,L_2\in\Ll$. By the modularity of the function $\lVert\cdot\rVert$, we have 
    \begin{equation*}
        \begin{split}
             \rho_{\Cc}(L_1)+\rho_{\Cc}(L_2)&=\lVert L_1\rVert+\lVert L_2\rVert-\lVert L_1\cap\Cc\rVert-\lVert L_2\cap\Cc\rVert=\\
             &=\lVert L_1+L_2\rVert+\lVert L_1\cap L_2\rVert-(\lVert L_1\cap L_2\cap\Cc\rVert+\lVert (L_1\cap\Cc)+(L_2\cap\Cc)\rVert)=\\
             &\geq\lVert L_1+L_2\rVert+\lVert L_1\cap L_2\rVert-(\lVert L_1\cap L_2\cap\Cc\rVert+\lVert (L_1+L_2)\cap\Cc\rVert)=\\
             &=\rho_{\Cc}(L_1\cap L_2)+\rho_{\Cc}(L_1+L_2),
        \end{split}
    \end{equation*}
    where the inequality follows from $(L_1\cap\Cc)+(L_2\cap\Cc)\subseteq (L_1+L_2)\cap\Cc$.
\end{proof}

The next example clarifies our reason for explicitly considering $\mathcal{R}^n$ in the previous proposition.

\begin{example}
Let $R$ be a finite field $\F_q$, and let $\Cc\subseteq\F_q^n$ be a linear block code. It is well known that the dimension is a modular function from the set of vector subspaces of $\F_q^n$ to $\Z$. Therefore, $(\rho_{\Cc},\dim,\F_q^n)$ is a $\Z$-latroid by Proposition~\ref{proposition:latroidcode}. In this case, the rectangular subspaces of $\F_q^n$ are direct products of copies of $\F_q$ and $\{0\}$. In particular, the rectangular subspaces are in bijection with the subsets of $[n]$. Therefore, we can construct an associated matroid as in Remark~\ref{remark:latroidtomatroid}. This matroid coincides with the matroid that we usually associate to a code endowed with the Hamming metric.
\end{example}

We point out that modular functions and modular supports were defined independently in two different contexts. So, even though they are both called modular, they are not the same class of functions. However, there are situations where modular supports are also modular functions. For instance, we now show that a standard modular support is also a strictly increasing modular function on the lattice $\mathcal{R}^n$, if $R$ is a principal ideal ring. We start by considering the case where $R$ is a finite chain ring.

\begin{lemma}\label{lemma:localpirmodular}
Let $R$ be a finite chain ring, and let $\supp:R^n\rightarrow \Z^u$ be a standard support. Then:
\begin{enumerate}
\item $\supp(M_1)\vee\supp(M_2)=\supp(M_1+M_2)$ and $\supp(M_1)\wedge\supp(M_2)=\supp(M_1\cap M_2)$ for any $M_1,M_2\in\mathcal{R}^n$. In particular $(\{\supp(M): M\in\mathcal{R}^n\},\leq)$ is a finite lattice.
\item $\supp: \mathcal{R}^n\rightarrow\Z^u$ is a modular function, i.e., $\supp(M_1)+\supp(M_2)=\supp(M_1+ M_2)+\supp(M_1\cap M_2)$ for all $M_1, M_2\in\mathcal{R}^n$.
\end{enumerate}
\end{lemma}

\begin{proof}
From the definition of $\supp$, it follows that $\supp(M_1)\vee\supp(M_2)\leq \supp (M_1+M_2)$. For every $m\in M_1+M_2$ there exist $m_1\in M_1$ and $m_2\in M_2$, such that $m=m_1+m_2$. Therefore, $\supp(m)\leq \supp(m_1)\vee\supp(m_2)\leq \supp(M_1)\vee\supp(M_2)$ and so $\supp(M_1)\vee\supp(M_2)=\supp(M_1+ M_2)$.

Since $M_1\cap M_2\subseteq M_1$ and $M_1\cap M_2\subseteq M_2$, then $\supp(M_1\cap M_2)\leq \supp(M_1)\wedge \supp(M_2)$. Fix $i\in[u]$. Since $\supp$ is standard, there exist $m_1=(0,\dots,0,(m_1)_i,0,\dots,0)\in M_1$ and $m_2=(0,\dots,0,(m_2)_i,0,\dots,0)\in M_2$ such that $\supp(m_1)_i=\supp(M_1)_i$ and $\supp(m_2)_i=\supp(M_2)_i$. Let $\alpha$ be a generator of the maximal ideal of $R$. There exist $r_1,r_2$ invertible elements and $k_1,k_2$ such that $(m_1)_i=r_1\alpha^{k_1}$ and $(m_2)_i=r_2\alpha^{k_2}$. Assume without loss of generality that $k_1\leq k_2$. Then $\supp(m_1)_i\geq \supp(m_2)_i$, so $m_2=r_1^{-1}r_2\alpha^{k_2-k_1}m_1\in M_1\cap M_2$. Hence $\supp(M_1\cap M_2)_i\geq \supp(m_2)_i=\supp(M_2)_i$, therefore $\supp(M_1\cap M_2)_i=\supp(M_1)_i\wedge \supp(M_2)_i$. We conclude, since $\supp(M_1)+\supp(M_2)=\supp(M_1)\wedge \supp(M_2)+\supp(M_1)\vee \supp(M_2)$.
\end{proof}

\begin{proposition}\label{proposition:pirmodular}
Let $R$ be a principal ideal ring and let $\supp:R^n\rightarrow \Z^u$ be a standard modular support. Then $\supp: \mathcal{R}^n\rightarrow\Z^u$ is a modular function.
\end{proposition}

\begin{proof}
By~Proposition~\ref{proposition:supportsplits}, $R^n=R_1^n\times\dots\times R_{\ell}^n$ with $R_1,\dots,R_{\ell}$ finite chain rings and $\supp=\supp_1\times\dots\supp_\ell$, where $\supp_i=R_i^n\rightarrow \Z^{u_i}$ is a standard modular support for all $i\in[\ell]$. If $M\in\mathcal{R}^n$, then $M=M_1\times\dots\times M_{\ell}$ with $M_i\in\mathcal{R}_i^n$, We conclude by applying Lemma~\ref{lemma:localpirmodular} to each $\supp_i$.
\end{proof}

While in Lemma~\ref{lemma:localpirmodular} we do not require the support to be modular, Proposition~\ref{proposition:pirmodular} does not hold in general without this assumption, as one can see in the next example.

\begin{example}\label{example:modularitypir}
Consider the ring $\Z_6$ endowed with the Hamming support and let $M_1=(2)$ and $M_2=(3)$. Then, we obtain $2=\supp(M_1)+\supp(M_2)\neq\supp(M_1+M_2)+\supp(M_1\cap M_2)=1$.
\end{example}

\begin{corollary}
Let $R$ be a principal ideal ring and let $\supp:R^n\rightarrow \Z^u$ be a standard modular support. The associated weight function is modular on $\mathcal{R}^n$.
\end{corollary}

\begin{proof}
The thesis follows from Proposition~\ref{proposition:pirmodular}, since the direct sum of modular functions is modular.
\end{proof}

The next example shows that the weight function associated to a modular support is not necessarily modular on $\mathcal{M}(R^n)$.

\begin{example}
Let $R=\F_2$ and let $\supp$ be the Hamming support. Let $M_1,M_2\in\mathcal{M}(\F_2^2)$, $M_1=\langle(1,1)\rangle$ and $M_2=\langle(0,1)\rangle$. Then $\wt(M_1)=2$, $\wt(M_2)=1$, $\wt(M_1+M_2)=2$, and $\wt(M_1\cap M_2)=0$.
\end{example}

Notice that there are standard supports that are not strictly increasing functions. For example, the Hamming support on $\Z_6$ is not strictly increasing, as $(1)\supsetneq(2)$ but $\wt(1)=\wt(2)=1$. However, in the following proposition we show that all standard modular supports on principal ideal rings are strictly increasing.

\begin{proposition}\label{propostion:pirstrictlyinc}
    Let $R$ be a principal ideal ring and let $\supp:R^n\rightarrow \Z^u$ be a standard modular support. Then, $\supp: \mathcal{R}^n\rightarrow\Z^u$ is strictly increasing. 
\end{proposition}

\begin{proof}
    As in the case of Proposition~\ref{proposition:pirmodular}, it suffices to prove the result for finite chain rings. So, assume that $R$ is a finite chain ring with maximal ideal generated by $\alpha$, and let $M_1<M_2$ be two rectangular submodules of $R^n$. Then, there exist $i\in[n]$, $m_1=(0,\dots,0,\alpha^{t_1},0,\dots,0)\in M_1$, and $m_2=(0,\dots,0,\alpha^{t_2},0,\dots,0)\in M_2$ with $t_2<t_1$, such that $\supp(M_1)_i=\supp(m_1)_i$ and $\supp(M_2)_i=\supp(m_2)_i$. Since $m_1\in M_2$, we have that $\supp(M_1)_i\leq\supp(M_2)_i$. Suppose now that they are equal. On the one side, since $\supp$ is a modular support, there is $r\in R$ such that $\supp((0,\dots,0,\alpha^{t_2}-r\alpha^{t_1},0,\dots,0))_i<\supp((0,\dots,0,\alpha^{t_2},0,\dots,0))_i$. On the other side, since $R$ is a local ring, we have that $1-r\alpha^{t_2-t_1}$ is an invertible element, and so $\supp((0,\dots,0,\alpha^{t_1}-r\alpha^{t_2},0,\dots,0))_i=\supp((0,\dots,0,\alpha^{t_1},0,\dots,0))_i$. This is a contradiction, and therefore we conclude that $\supp(M_1)_i<\supp(M_2)_i$, and so $\supp(M_1)<\supp(M_2)$.
\end{proof}

A standard modular support on a principal ideal ring defines a strictly increasing modular function on $\mathcal{R}^n$
by Proposition~\ref{proposition:pirmodular} and Proposition~\ref{propostion:pirstrictlyinc}. However, the same support is not in general a modular function on $\mathcal{M}(R^n)$. In particular, if we define $\rho_{\Cc}$ as $\rho_{\Cc}(M)=\supp(M)-\supp(M\cap\Cc)$ for $M\in\mathcal{R}^n$, then the triple $(\rho_{\Cc},\supp,\mathcal{R}^n)$ may not be a $\Z^u$-latroid. However, given a standard modular support, we can construct a $\Z^u$-latroid as follows. 

\begin{definition}\label{def:latroid_supp}
Let $\supp$ be a standard modular support and define 
$\rho_{\Cc}^{\supp}:\mathcal{M}(R^n)\rightarrow\Z^u$ as
\begin{equation*}
    \rho_{\Cc}^{\supp}(M)=\supp(M)-\supp\left(\min\{L\in\mathcal{R}^n:M\cap\Cc\subseteq L\}\right)
\end{equation*}
for $M\in\mathcal{M}(R^n)$.
\end{definition}

Let $\bar\Cc$ be the smallest rectangular submodule that contains $\Cc$. It is easy to verify that $\bar\Cc=\pi_1(\Cc)\times\dots\times\pi_{n}(\Cc)$, where $\pi_i$ is the canonical projection on the $i$-th entry, and that $M\cap\bar\Cc=\min\{L\in\mathcal{R}^n:M\cap\Cc\subseteq L\}$. Following the proof of Proposition~\ref{proposition:latroidcode}, one can prove that $(\rho_{\Cc}^{\supp},\supp,\mathcal{R}^n)$ is a latroid.

\subsection{Chain support, Tutte polynomial, and weight distribution}

The weight enumerator is a central and widely studied invariant in coding theory \cite[Chapter VI]{van1971coding}. It captures many interesting properties of a code, e.g., it may be used to better understand the decoding properties of the code. For instance, the weight enumerator of a binary code allows us to estimate the probability that a received codeword is closer to a different codeword compared to the actual transmitted codeword \cite[Section 3]{jurrius2009codes}.
In this section, we propose a finer definition of weight enumerator for linear codes over rings. For an $R$-linear code endowed with the chain support (see Definition~\ref{definition:chainsupport}), we show how to recover the weight enumerator from a suitable associated latroid. 

\begin{definition}
The {\bf homogeneous weight enumerator}\index{weight enumerator!homogeneous} of an $R$-linear code $\Cc\subseteq R^n$ is the
polynomial
$$W_{\Cc}(x,y)=\sum_{c\in\Cc} x^{\wt(c)}y^{\wt(R^n)-\wt(c)}.$$
\end{definition}

The homogeneous weight enumerator can also be written as
\begin{equation*}
    W_{\Cc}(x,y)=\sum_{w=0}^{\wt(R^n)}A_wx^{w}y^{\wt(R^n)-w},
\end{equation*}
where $$A_w=\lvert\{c\in\Cc:\wt(c)=w\}\rvert.$$ 
The list $A_0,\dots,A_{\wt(R^n)}$ is called the {\bf weight distribution} of $\Cc$ and is an invariant of the code. Notice that, if $R$ is a field and the weight is the Hamming weight, we have $\wt(R^n)=n$ and we recover the classical definition of weight enumerator. In the more general case of standard support, it is important to keep track of what happens in each component. For this reason, we introduce a finer version of the weight enumerator, which we call support enumerator. 

\begin{definition}
For a support $\supp:R^n\rightarrow \Z^u$, we define the {\bf support enumerator}\index{weight enumerator!complete} as
\begin{equation*}
    W_{\Cc}(\mathbf{x},\mathbf{y})=\sum_{c\in\Cc}\mathbf{x}^{\supp(c)}\mathbf{y}^{\overline{\supp(c)}} .
\end{equation*}
where $\mathbf{x}^{\supp(c)}=\prod_{i=1}^{u}x_i^{\supp(c)_i}$ and $\overline{\supp(c)}=\supp(R^n)-\supp(c)$. 
\end{definition}

One can recover the homogeneous weight enumerator from the support enumerator by setting $x_1=\dots=x_u=x$ and $y_1=\dots=y_u=y$. The next lemma follows by direct computation.

\begin{lemma}\label{lemma:weightenumeratorprod}
    Let $R=R_1\times\dots\times R_{\ell}$ be a principal ideal ring, $\Cc=\Cc_1\times\dots\times\Cc_{\ell}$ be an $R$-linear code, and $\supp=\supp_1\times\dots\times\supp_{\ell}$ be a modular support. Then
    \begin{equation*}
        W_{\Cc}(\mathbf{x}_1,\dots,\mathbf{x}_{\ell},\mathbf{y}_1,\dots,\mathbf{y}_{\ell})=\prod_{i=1}^{\ell}W_{\Cc_i}(\mathbf{x}_i,\mathbf{y}_{i}).
    \end{equation*}
\end{lemma}

In addition to the weight enumerator, we are also interested in the generalized weight enumerator.

\begin{definition}\label{definition:genweienum}
Let $\Cc$ be an $R$-linear code. For $0\leq r\leq \lambda(\Cc)$, the {\bf $r$-th generalized weight enumerator}\index{weight enumerator!generalized} is 
\begin{equation*}
    W^{(r)}_{\Cc}(x,y)=\sum_{r=0}^{M(\Cc)}A^{(r)}_{w}x^{\wt(R^n)-w}y^{w},
\end{equation*}
where $A^{(r)}_{w}=\lvert\{\Dd\subseteq\Cc:\lambda(\Dd)=r\text{ and }\wt(\Dd)=w\}\lvert$.
\end{definition}

Similarly to how the weight enumerator captures the weight distribution, the generalized weight enumerator captures the generalized weights. Indeed, we have
\begin{equation*}
    \bar d_r(\Cc)=\min\{w: A^{(j)}_w\neq 0 \text{ for some } j\geq r\}
\end{equation*}
for $1\leq r\leq \lambda(\Cc)$.

The {\bf Tutte polynomial}\index{Tutte polynomial} was introduced in~\cite{tutte1947ring,tutte1954contribution} for graphs and the definition was generalized to matroids in~\cite{crapo1969tutte}. For a matroid $(E,\rho)$ it is defined as
\begin{equation*}
    T(\rho,x,y)=\sum_{A\subseteq E}(x-1)^{\rho(E)-\rho(A)}(y-1)^{\lvert A\rvert-\rho(A)}.
\end{equation*}
The {\bf Tutte-Whitney rank generating function} is obtained from the Tutte polynomial via a change of variables
\begin{equation*}
    R(\rho,x,y)=T(\rho,x+1,y+1)=\sum_{A\subseteq E}x^{\rho(E)-\rho(A)}y^{\lvert A\rvert-\rho(A)}.
\end{equation*}
In~\cite{vertigan2004latroids}, Vertigan extends the definition of Tutte-Whitney rank generating function\index{Tutte-Whitney rank generating function} to latroids as follows.

\begin{definition}
    The {\bf weighted Tutte-Whitney rank generating function} of a $\Z^u$-latroid $(\rho,\lVert\cdot\rVert,\mathcal{L})$ with $\mathcal{L}\subseteq \Z^u$ is
    \begin{equation*}
    R(\rho,\lVert\cdot\rVert,\mathcal{L},\mathbf{x},\mathbf{y},\mathbf{u},\mathbf{v})
    =\sum_{M\in\mathcal{L}}\mathbf{x}^M\mathbf{y}^{M^\perp}\mathbf{u}^{\rho(1_\mathcal{L})-\rho(M)}\mathbf{v}^{\lVert M\rVert-\rho(M)}.
    \end{equation*}
\end{definition}

Since $\mathcal{L}$ is a sublattice of $\Z^u$, we have $M^\perp=\overline{M}=1_{\mathcal{L}}-M$. Moreover, observe that the weighted Tutte-Whitney rank generating function fully determines the function 
\begin{equation}\label{eqn:Rprime}
R'(\rho,\lVert\cdot\rVert,\mathcal{L},\mathbf{x},\mathbf{z},\mathbf{y},\mathbf{u},\mathbf{v})
=\sum_{M\in\mathcal{L}}\mathbf{x}^M\mathbf{z}^{\tilde M-M}\mathbf{y}^{M^\perp}\mathbf{u}^{\rho(1_\mathcal{L})-\rho(M)}\mathbf{v}^{\lVert M\rVert-\rho(M)},
\end{equation}
where $\tilde M=(M+(1,\dots,1))\wedge 1_{\mathcal{L}}$. Notice that when $\mathcal{L}= \{0,1\}^n$, then $\tilde M=(1,\dots,1)$. Conversely, (\ref{eqn:Rprime}) determines the weighted Tutte-Whitney rank generating function via $$R(\rho,\lVert\cdot\rVert,\mathcal{L},\mathbf{x},\mathbf{y},\mathbf{u},\mathbf{v})=R'(\rho,\lVert\cdot\rVert,\mathcal{L},\mathbf{x},\mathbf{1},\mathbf{y},\mathbf{u},\mathbf{v}).$$

Given an $R$-linear code $\Cc\subseteq R^n$, we associate a $\Z$-latroid to it as follows.
  
\begin{lemma}\label{lemma:latroidz}
Let $R$ be a finite chain ring, and let $\Cc$ be an $R$-linear code. Let $$\mathcal{L}_R=\{\supp(M):M\in\mathcal{R}^n\}$$ and $$\rho_{\Cc}(\supp(M))=\lvert \supp(M)\rvert-\lambda(M\cap\Cc),$$ where $\supp$ denotes the chain support. Then $\rho_{\Cc}$ is well defined, and the triple $(\rho_{\Cc},\lvert\cdot\rvert,\mathcal{L}_R)$ is a $\Z$-latroid.
\end{lemma}
    
\begin{proof}    
By Lemma~\ref{lemma:localpirmodular}, $\mathcal{L}_R$ is a sublattice of $\Z^n$. Let $\alpha$ be a generator of the maximal ideal of $R$ and let $k=\min\{i\in\Z_{\geq 1}:\alpha^i=0\}$. Then $\supp(M)=(i_1,\dots,i_n)$ if and only if $M=(\alpha^{k-i_1})\times\dots\times(\alpha^{k-i_n})$. 
Therefore, $\mathcal{L}_R$ is in one-to-one correspondence with $\mathcal{R}^n$, so $\rho_{\Cc}$ is well defined. Moreover, $\rho_{\Cc}(M)=\lambda(M)-\lambda(M\cap\Cc)$. The thesis now follows from~\cite[Lemma 5.9]{vertigan2004latroids}.
\end{proof}

\begin{definition}
Let $R$ be a finite chain ring, and let $\Cc\subseteq R^n$ be an $R$-linear code. Let $$\mathcal{L}_R=\{\supp(M):M\in\mathcal{R}^n\}$$ and $$\rho_{\Cc}(\supp(M))=\lvert \supp(M)\rvert-\lambda(M\cap\Cc).$$ The triple $(\rho_{\Cc},\lvert\cdot\rvert,\mathcal{L}_R)$ is the {\bf chain support latroid} associated to $\Cc$. 
\end{definition}

\begin{remark}\label{rmk:samelatroid}
For $M\in\mathcal{R}^n$, let $\rho^{\lambda}_{\Cc}(M)=\lambda(M)-\lambda(M\cap\Cc)$.
By Lemma~\ref{lemma:latroidz}, the $\Z$-latroid $(\rho^{\lambda}_{\Cc},\lambda,\mathcal{R}^n)$ defined by Vertigan in~\cite[Example~5.8]{vertigan2004latroids} and the chain support latroid $(\rho_{\Cc},\lvert\cdot\rvert,\mathcal{L}_R)$ may be identified via the bijection $\iota:\mathcal{R}^n\rightarrow\mathcal{L}_R$ given by $\iota(M)=\supp(M)$. In fact, $\iota$ satisfies $\rho^{\lambda}_{\Cc}=\rho_{\Cc}\circ\iota$ and $\lambda=\lvert\cdot\rvert\circ\iota$.
\end{remark}

The generalized weights of an $R$-linear code $\Cc$ according to Definition~\ref{definition:genweiring2} coincide with the generalized weights of its chain support latroid.
    
\begin{proposition}\label{prop:gwtsfromlatroid}
Let $R$ be a finite chain ring, and let $\Cc$ be an $R$-linear code. Then
    \begin{equation*}
        \bar d_r(\Cc)=d_r(\rho_{\Cc},\lvert\cdot\rvert,\mathcal{L}_R),
    \end{equation*}
for $1\leq r\leq \lambda(\Cc)$.
\end{proposition}
    
\begin{proof}
By definition and using Remark~\ref{rmk:samelatroid} we have 
$$d_r(\rho_{\Cc},\lvert\cdot\rvert,\mathcal{L}_R)=\min_{\Dd\in\mathcal{R}^n}\{\lvert\supp(\Dd)\rvert: \lambda(\Dd\cap\Cc)\geq r\}\leq\min_{\Dd\subseteq\Cc}\{\lvert\supp(\Dd)\rvert: \lambda(\Dd)\geq r\}=\bar d_r(\Cc).$$
To prove the reverse inequality, let $\Dd$ be a submodule of $\Cc$ that realizes $\bar d_r(\Cc)$, i.e., $\Dd\subseteq\Cc$ has $\lambda(\Dd)\geq r$ and $\bar d_r(\Cc)=\lvert\supp(\Dd)\rvert$. Let $\bar \Dd\in\mathcal{R}^n$ be the smallest element that contains $\Dd$. Then $\lvert\supp(\bar \Dd)\rvert=\lvert\supp(\Dd)\rvert$ and $\lambda(\bar \Dd)\geq \lambda(\Dd)$. Therefore,
\begin{equation*}
    \bar d_r(\Cc)=\lvert\supp(D)\rvert= \lvert\supp(\bar D)\rvert\geq d_r(\rho_{\Cc},\lvert\cdot\rvert,\mathcal{L}_R).\qedhere
\end{equation*}
\end{proof}
    
In the next lemma we recall a fact from commutative algebra, that we will use in the proof of Theorem~\ref{theorem:TutteWhitney}.
    
\begin{lemma}\label{lemma:cardmodulefcr}
Let $R$ be a finite chain ring, and let $M$ be a finitely generated $R$ module. Then $$\lvert M\rvert=\lvert R/(\alpha)\rvert^{\lambda(M)}.$$
\end{lemma}
    
\begin{proof}
By the definition of length of a module, there exists a chain of modules
$$M=M_0\supsetneq M_1\supsetneq \dots \supsetneq M_{\lambda(M)},$$
that is a composition series, i.e., $M_i/M_{i+1}$ is a nonzero simple $R$-module for $0\leq i<\lambda(M)$, see \cite[Theorem 2.13]{eisenbud2013commutative}. A simple $R$-module is isomorphic to $R/J$, where $J$ is a maximal ideal of $R$. Since $R$ is a local ring, we conclude that $M_i/M_{i+1}\cong R/(\alpha)$ for $0\leq i<\lambda(M)$. We conclude by induction on the length of the composition series.
\end{proof}
    
We can now show that the support enumerator of a code $\Cc$ is determined by the weighted Tutte-Whitney rank generating function of the associated chain support latroid. The proof of the following theorem extends the proof of \cite[Theorem 9.4]{vertigan2004latroids}.
    
\begin{theorem}\label{theorem:TutteWhitney}
Let $R$ be a finite chain ring, and let $\Cc\subseteq R^n$ be an $R$-linear code. The Tutte-Whitney rank generating function of  $(\rho_{\Cc},\lvert\cdot\rvert,\mathcal{L}_R)$ determines the support enumerator of $\Cc$ according to the formula
\begin{equation*}
W_{\Cc}(\mathbf{x},\mathbf{y})=R'\left(\rho_{\Cc},\lVert\cdot\rVert,\mathcal{L}_R,\mathbf{x},\frac{\mathbf{y}-\mathbf{x}}{\mathbf{y}},\mathbf{y},\lvert R/(\alpha)\rvert,1\right).
\end{equation*}
\end{theorem}
    
\begin{proof}
For each $A\in\Z^n$, let $C_A=\{c\in C: \supp(c)\leq A\}$, and let $A_i=A-e_i$ for all $i\in[n]$.
\begin{equation}\label{equation:lat1}
\begin{split}
n_{C}(A):=&\lvert \{ c\in C:\supp(c)=A\}\rvert=\lvert C_A\rvert -\lvert \bigcup_{i=1}^n C_{A_i}\rvert=\\
=&\lvert C_A\rvert-\sum _{k=1}^{n}(-1)^{k+1}\left(\sum_{1\leqslant i_{1}<\cdots <i_{k}\leqslant n}|C_{A_{i_{1}}}\cap \cdots \cap C_{A_{i_{k}}}|\right)=\\
=&\sum_{A-(1,\dots,1)\leq B\leq A}(-1)^{\lvert A\rvert -\lvert B\rvert} \lvert C_B\rvert.
\end{split}
\end{equation}
One can check by direct computation that
    \begin{equation}\label{equation:lat2}
        \mathbf{y}^B(\mathbf{y}-\mathbf{x})^{(1,\dots,1)-B}=\sum_{B\subseteq A\subseteq (1,\dots,1)}(-1)^{\lvert A \rvert-\lvert B\rvert} \mathbf{x}^{A}\mathbf{y}^{ (1,\dots,1)-A},
    \end{equation}
for all $(0,\dots,0)\leq B\leq (1,\dots,1)$. Let $\tilde B=(B+(1,\dots,1))\wedge \supp(R^n)$. We have
    \begin{equation*}
        \begin{split}
            \sum_{c\in C}\mathbf{x}^{\supp(c)}\mathbf{y}^{\overline{\supp(c)}}=&\sum_{A\in \mathcal{L}} n_C(A)\mathbf{x}^A\mathbf{y}^{\overline{A}}=\sum_{A\in\mathcal{L}}\left(\sum_{A-(1,\dots,1)\leq B\leq A}(-1)^{\lvert A\rvert -\lvert B\rvert}  \lvert C_B\rvert\right)\mathbf{x}^A\mathbf{y}^{\overline{A}}\\
            =&\sum_{B\in \mathcal{L}}\left(\lvert C_B\rvert\sum_{B\leq A\leq \tilde B}(-1)^{\lvert A\rvert -\lvert B\rvert}  \mathbf{x}^A\mathbf{y}^{\overline{A}}\right)=\\
            =&\sum_{B\in \mathcal{L}}\left(\lvert C_B\rvert \mathbf{x}^{\tilde B-(1,\dots,1)}\mathbf{y}^{\supp(R^n)-\tilde B}\sum_{B\leq A\leq \tilde B}(-1)^{\lvert A\rvert -\lvert B\rvert}  \mathbf{x}^{A-\tilde B+(1,\dots,1)}z^{\tilde B -A}\right)=\\
            =&\sum_{B\in \mathcal{L}}\lvert C_B\rvert \mathbf{x}^{\tilde B-(1,\dots,1)}\mathbf{y}^{\supp(R^n)-\tilde B}\mathbf{x}^{B-\tilde B +(1,\dots,1)}(\mathbf{y}-\mathbf{x})^{\tilde B-B}=\\
            =&\sum_{B\in \mathcal{L}}\lvert C_B\rvert \mathbf{x}^{B}\mathbf{y}^{\supp(R^n)-\tilde B}(\mathbf{y}-\mathbf{x})^{\tilde B-B}=\sum_{B\in \mathcal{L}}\lvert C_B\rvert \mathbf{x}^{B}\mathbf{y}^{\overline{B}}\left(\frac{\mathbf{y}-\mathbf{x}}{\mathbf{y}}\right)^{\tilde B-B},
        \end{split}
    \end{equation*}
where in the second equality we used Equation \eqref{equation:lat1} and in the second to last we used Equation \eqref{equation:lat2}.
Since $R$ is a finite chain ring and $\lvert B\rvert-\rho_{\Cc}(B)$ is the length of $\Cc_B$, by Lemma~\ref{lemma:cardmodulefcr} we have $\lvert \Cc_B\rvert=\lvert R/(\alpha)\rvert^{\lvert B\rvert-\rho_{\Cc}(B)}$. Combining these results, we obtain
    \begin{equation*}
        W_{\Cc}(\mathbf{x},\mathbf{y})=\sum_{B\in \mathcal{L}}\lvert R/(\alpha)\rvert^{\lvert B\rvert-\rho_{\Cc}(B)}\mathbf{x}^{B}\mathbf{y}^{\overline{B}}\left(\frac{\mathbf{y}-\mathbf{x}}{\mathbf{y}}\right)^{\tilde B-B}=R'\left(\rho_{\Cc},\lVert\cdot\rVert,\mathcal{L}_R,\mathbf{x}
            ,\frac{\mathbf{y}-\mathbf{x}}{\mathbf{y}},\mathbf{y},\lvert R/(\alpha)\rvert,1\right).\qedhere
    \end{equation*}
\end{proof}

\begin{remark}
In Theorem~\ref{theorem:TutteWhitney} we prove that the support enumerator, hence the homogenoeus weight enumerator, can be obtained from $R'(\rho,\lVert\cdot\rVert,\mathcal{L},\mathbf{u},\mathbf{z},\mathbf{v},\mathbf{x},\mathbf{y})$. As we stated above, this function is determined by the weighted Tutte-Whitney rank generating function. Therefore, one can also write the weight enumerator in terms of the Tutte-Whitney rank generating function, but the formula is not as concise.
\end{remark}

Proposition~\ref{prop:gwtsfromlatroid} and Theorem~\ref{theorem:TutteWhitney} can be extended to codes over a finite principal ideal ring $R$. Recall that every submodule $M$ of $R^n$ decomposes as direct product $M_1\times\dots\times M_{\ell}$ where $M_i$ is a submodule of $R_i^n$ for each $i\in[\ell]$, moreover $M\in\mathcal{R}^n$ if and only if $M_i\in\mathcal{R}_i^n$ for each $i\in[\ell]$. 

\begin{lemma}
Let $R$ be a finite principal ideal ring and let $\Cc$ be an $R$-linear code. Let $\supp_i$ be the chain support on $R_i$ for $i\in[\ell]$ and let $\supp=\supp_1\times\dots\times\supp_{\ell}$. Let $M=M_1\times\dots\times M_{\ell}\in\mathcal{R}^n=\mathcal{R}_1^n\times\dots\times\mathcal{R}_{\ell}^n$ and define $$\rho_{\Cc}(\supp(M))=(\rho_{\Cc_1}(\supp(M_1)),\dots,\rho_{\Cc_\ell}(\supp(M_\ell))\in\Z^\ell.$$ Then $\supp$ is a modular support on $R^n$ and the triple $(\rho_{\Cc},\lvert\cdot\rvert,\mathcal{L}_R)$ is a $\Z^{\ell}$-latroid.
\end{lemma}

\begin{proof}
The fact that $\supp$ is a modular support on $R^n$ follows from~\cite[Proposition~3.17]{gorla2022generalized}.
Moreover, $\rho_{\Cc}$ is bounded increasing and submodular if and only if $\rho_{\Cc_i}$ is bounded increasing and submodular for all $i\in[\ell]$. We conclude by Lemma~\ref{lemma:latroidz}.
\end{proof}

\begin{definition}
Let $R$ be a finite principal ideal ring and let $\Cc\subseteq R^n$ be an $R$-linear code. 
Let $\supp_i$ be the chain support on $R_i$ for $i\in[\ell]$ and let $\supp=\supp_1\times\dots\times\supp_{\ell}$.
Let $$\mathcal{L}_R=\{\supp(M):M\in\mathcal{R}^n\}=\mathcal{L}_{R_1}\times\dots\times\mathcal{L}_{R_\ell}$$ and 
$$\rho_{\Cc}(\supp(M))=(\rho_{\Cc_1}(\supp(M_1)),\dots,\rho_{\Cc_\ell}(\supp(M_\ell))\in\Z^\ell.$$
The triple $(\rho_{\Cc},\lvert\cdot\rvert,\mathcal{L}_R)$ is the {\bf chain support latroid} associated to $\Cc$.
\end{definition}

Over a finite principal ideal ring, the Tutte-Whitney rank generating function of the chain support latroid of a code determines its weight enumerator and support enumerator. 

\begin{theorem}\label{thm:TutteWhitneyPIR}
    Let $R$ be a finite principal ideal ring and let $\Cc\subseteq R^n$ be an $R$-linear code. The Tutte-Whitney rank generating function of  $(\rho_{\Cc},\lvert\cdot\rvert,\mathcal{L}_R)$ determines the support enumerator of $\Cc$. In particular,
    \begin{equation*}
        W_{\Cc}(\mathbf{x},\mathbf{y})=R'\left(\rho_{\Cc},\lVert\cdot\rVert,\mathcal{L}_{R},\mathbf{x}
            ,\frac{\mathbf{y}-\mathbf{x}}{\mathbf{y}},\mathbf{y},\lvert R/(\alpha_1)\rvert,\dots,\lvert R/(\alpha_{\ell})\rvert,\mathbf{1}\right)
    \end{equation*}
\end{theorem}

\begin{proof}
    By Lemma~\ref{lemma:weightenumeratorprod} and Theorem~\ref{theorem:TutteWhitney} we have
    \begin{equation*}
        W_{\Cc}(\mathbf{x}_1,\dots,\mathbf{x}_{\ell},\mathbf{y}_1,\dots,\mathbf{y}_{\ell})=\prod_{i=1}^{\ell}W_{\Cc_i}(\mathbf{x}_i,\mathbf{y}_{i})=\prod_{i=1}^{\ell}R'\left(\rho_{\Cc_i},\lVert\cdot\rVert,\mathcal{L}_{R_i},\mathbf{x}_i
            ,\frac{\mathbf{y}_i-\mathbf{x}_i}{\mathbf{y}_i},\mathbf{y}_i,\lvert R/(\alpha_i)\rvert,1\right).
    \end{equation*}
    Since 
    \begin{equation*}
        \begin{split}
            R'(\rho_{\Cc},\lVert\cdot\rVert,\mathcal{L}_R,&\mathbf{x}_1,\dots,\mathbf{x}_{\ell},\mathbf{z}_1,\dots,\mathbf{z}_{\ell},\mathbf{y}_1,\dots,\mathbf{y}_{\ell},u_1,\dots,u_\ell,v_1\dots,v_\ell)=\\
            &=\prod_{i=1}^{\ell}R(\rho_{\Cc_i},\lVert\cdot\rVert,\mathcal{L}_{R_i},\mathbf{x}_i,\mathbf{z}_i,\mathbf{y}_i,u_i,v_i).
        \end{split}
    \end{equation*}
    we conclude.
\end{proof}

\subsection{Chain support latroid and generalized weights}

In this short section we show that, over a finite principal ideal ring, the generalized weights of the chain support latroid of a code coincide with the generalized weights of the code. We start with two observations on generalized weights.

\begin{lemma}\label{lemma:decomposition}
Let $R$ be a finite principal ideal ring and write $R=R_1\times\dots\times R_{\ell}$ with $R_i$ finite chain ring for all $i\in[\ell]$. Let $\Cc$ be an $R$-linear code and write $\Cc=\Cc_1\times\dots\times\Cc_{\ell}$ with $\Cc_i$ an $R_i$-linear code. Then $$\bar d_r(\Cc)=\min\left\{\sum_{i=1}^\ell \bar d_{r_i}(\Cc_i): r_1+\dots+r_{\ell}\geq r,\; 0\leq r_i\leq\lambda(\Cc_i) \text{ for all }i  \right\}.$$
\end{lemma}

\begin{proof}
Notice that every subcode $\Dd\subseteq\Cc$ is of the form $\Dd=\Dd_1\times\dots\times\Dd_{\ell}$ with $\Dd_i\subseteq\Cc_i$ a subcode for $i\in[\ell]$. Moreover, $\lvert\supp(\Dd)\rvert=\sum_{i=1}^{\ell}\lvert\supp(\Dd_i)\rvert$ and $\lambda(\Dd)=\sum_{i=1}^{\ell}\lambda(\Dd_i)$. The thesis now follows from the definition of generalized weights.
\end{proof}

\begin{lemma}\label{lemma:decompositionlatroids}
Let $R$ be a finite principal ideal ring and write $R=R_1\times\dots\times R_{\ell}$ with $R_i$ finite chain ring for all $i\in[\ell]$. Let $\Cc$ be an $R$-linear code and write $\Cc=\Cc_1\times\dots\times\Cc_{\ell}$ with $\Cc_i$ an $R_i$-linear code. Let $(\rho_{\Cc},\lvert\cdot\rvert,\mathcal{L}_R)$ be the chain support latroid associated to $\Cc$ and let  $(\rho_{\Cc_i},\lvert\cdot\rvert_i,\mathcal{L}_{R_i})$ be the chain support latroid associated to $\Cc_i$ for $i\in[\ell]$.
Then $$d_r(\rho_{\Cc},\lvert\cdot\rvert,\mathcal{L}_R)=d_r((\rho_{\Cc_1},\lvert\cdot\rvert_1,\mathcal{L}_{R_1})\oplus\dots\oplus(\rho_{\Cc_\ell},\lvert\cdot\rvert_{\ell},\mathcal{L}_{R_\ell})).$$   
\end{lemma}

\begin{proof}
For brevity, we denote by $(\rho,\lVert\cdot\rVert,\Ll)$ the direct sum $(\rho_{\Cc_1},\lvert\cdot\rvert_1,\mathcal{L}_{R_1})\oplus\dots\oplus(\rho_{\Cc_\ell},\lvert\cdot\rvert_{\ell},\mathcal{L}_{R_\ell})$. 
The latroids
$(\rho_{\Cc},\lvert\cdot\rvert,\mathcal{L}_R)$ and $(\rho,\lVert\cdot\rVert,\Ll)$ are over the same lattice $\mathcal{L}_R=\Ll=\mathcal{L}_{R_1}\times\dots\times\mathcal{L}_{R_\ell}$ and their rank and length functions are related via
$$\rho=\lvert\rho_{\Cc}\rvert \text{ and } \lVert\cdot\rVert=\sum_{i=1}^\ell\lvert\cdot\rvert_i.$$ The thesis follows from the definition of generalized weights of a latroid.
\end{proof}

\begin{lemma}\label{lemma:gwtsdirectsum}
Let $(\rho_1,\lVert\cdot\rVert_1,\Ll_1)$ and $(\rho_2,\lVert\cdot\rVert_2,\Ll_2)$ be $A$-latroids. Then 
$$d_a((\rho_1,\lVert\cdot\rVert_1,\Ll_1)\oplus (\rho_2,\lVert\cdot\rVert_2,\Ll_2))=\min\{d_{a_1}(\rho_1,\lVert\cdot\rVert_1,\Ll_1)+d_{a_2}(\rho_2,\lVert\cdot\rVert_2,\Ll_2): a_1,a_2\in A, a_1+a_2\geq a\}$$
for all $a\in A$. 
\end{lemma}

\begin{proof}
Denote by $(\rho,\lVert\cdot\rVert,\Ll)$ the direct sum $(\rho_1,\lVert\cdot\rVert_1,\Ll_1)\oplus (\rho_2,\lVert\cdot\rVert_2,\Ll_2)$.
Every $L=(L_1,L_2)\in\Ll=\Ll_1\times\Ll_2$ has $\lVert L\rVert=\lVert L_1\rVert_1+\lVert L_2\rVert_2$ and $\lVert L\rVert-\rho(L)=\lVert L_1\rVert_1-\rho_1(L_1)+\lVert L_2\rVert_2-\rho_2(L_2)$. If $L\in\Ll$ is such that $d_a(\rho,\lVert\cdot,\rVert,\Ll)=\lVert L\rVert$ and $\rho(L)-\lVert L\rVert\geq a$, then 
\end{proof}

The next result follows by combining Proposition~\ref{prop:gwtsfromlatroid} with  Lemma~\ref{lemma:decomposition}, Lemma~\ref{lemma:decompositionlatroids}, and Lemma~\ref{lemma:gwtsdirectsum}.

\begin{proposition}
Let $R$ be a finite principal ideal ring and let $\Cc$ be an $R$-linear code. Then
\begin{equation*}
\bar d_r(\Cc)=d_r(\rho_{\Cc},\lvert\cdot\rvert,\mathcal{L}_R)=d_r((\rho_{\Cc_1},\lvert\cdot\rvert,\mathcal{L}_{R_1})\oplus\dots\oplus(\rho_{\Cc_\ell},\lvert\cdot\rvert,\mathcal{L}_{R_\ell})),
\end{equation*}
for $1\leq r\leq \lambda(\Cc)$.
\end{proposition}

\subsection{Monomial ideals}

Given a modular support $\supp:R^n\rightarrow\Z^u$ we can associate to each nonzero $R$-linear code $\Cc$ a monomial ideal $I_{\Cc}$ defined as
\begin{equation*}
    I_{\Cc}=(\{\mathbf{x}^{\supp(c)}:c\in\Cc\setminus 0\})=(\{\mathbf{x}^{\supp(c)}:c\in\mathrm{Min}(\Cc)\})\subseteq S=\mathbb{K}[x_1,\ldots,x_u],
\end{equation*}
for an arbitrary field $\mathbb{K}$. We stress that all the results discussed in this section do not depend on the choice of $\mathbb{K}$. 

The ideal $I_\Cc$ can be recovered from the latroid $(\rho_\Cc^{\supp},\supp,\mathcal{R}^n)$ of Definition~\ref{def:latroid_supp}. Indeed, we have
\begin{equation*}
    \{\supp(c):c\in\mathrm{Min}(\Cc)\}=\{\supp(M):\supp(M)-\rho_{\Cc}^{\supp}(M)>0\text{ with }M\in\mathcal{R}^n\text{ and }M\text{ minimal}\}.
\end{equation*}
Notice that proceeding in this way we can always associate an ideal to a latroid. This correspondence is not bijective, since from different latroids we may obtain the same ideal. For instance, let $R$ be a principal ideal ring and let $(\rho_\Cc,\lvert\cdot\rvert,\mathcal{L}_R)$ be the chain support latroid. Then, the associated ideal is again $I_\Cc$.

From the ideal $I_{\Cc}$ one can recover some information on the code $\Cc$. For instance, in \cite[Theorem 4.4]{gorla2022generalized} the authors proved that the graded Betti numbers of the monomial ideal associated to a code determine its generalized weights according to Definition \ref{definition:genweiring}.
In \cite[Theorem 5.1]{johnsen2016generalization} the authors proved that the weight enumerator of an $\F_q$-linear block code is determined by the $\N_0$-graded Betti numbers associated with the $\N_0$-graded minimal free resolutions of the ideal of $I_{\Cc}$ and of its elongations, see \cite[Section 2]{johnsen2016generalization} for a definition. Here, we prove that the weight enumerator of an $\F_q$-linear code $\Cc\subseteq \F_q^n$ is determined by its $\N_0^n$-graded Betti numbers. An equivalent formula in terms of the nullity function of the associated matroid previously appeared as \cite[Proposition 3.1]{johnsen2016generalization}.

\begin{theorem}\label{theorem:bettifield}
Let $\Cc\subseteq\F_q^n$ be a linear code. Then,
\begin{equation*}
    A_w=\sum_{\lvert X\rvert=w,X\leq\{0,1\}^n}\sum_{Y\leq X}(-1)^{\lvert X\setminus Y\rvert}q^{\max\{i:\exists\, \beta_{i,Z}>0\text{ and }Z\leq Y\}},
\end{equation*}
for $w\in[n]$ and where $\{\beta_{i,Z}\}$ is the set of the $\N_0^n$-graded Betti numbers that appears in a minimal free resolution of $S/I_{\Cc}$. More precisely, $\beta_{i,X}$ the rank of the free module $S(-X)$ in homological position $i$. 
\end{theorem}

\begin{proof}
    By the inclusion–exclusion principle we obtain
    $$A_w=\sum_{\lvert X\rvert=w,X\leq\{0,1\}^n}\sum_{Y\leq X}(-1)^{\lvert X\setminus Y\rvert}\lvert\Cc_Y\rvert,$$
    where $\Cc_Y=\{c\in C: \supp(c)\leq Y\}$. The dimension of $\Cc_y$ is equal to the projective dimension of $S/I_{\Cc_Y}$. Fix a minimal free resolution of $S/I_{\Cc}$. The subresolution obtained by restricting to the direct summands $S(-X)$ with $X\leq Y$ is a minimal free resolution of $S/I_{\Cc_Y}$. Therefore, we have $\dim(\Cc_Y)=\max\{i:\exists\, \beta_{i,Z}>0\text{ and }Z\leq Y\}$. Since $\lvert \Cc_Y\rvert=q^{\dim(\Cc_Y)}$, we conclude.
\end{proof}
In the more general case of $R$-linear codes, the ideal $I_\Cc$ does not determine the weight enumerator, not even in the case of the chain support, as the next example shows.
\begin{example}
Consider the $\Z_4$-linear codes $\Cc_1$ and $\Cc_2$ in $\Z_4^3$ given by $$\Cc_1=\langle(2,1,0),(2,0,1)\rangle_{\Z_4}\text{ and }\Cc_2=\langle(2,1,0),(0,0,1)\rangle_{\Z_4}.$$  
With respect to the chain support, we obtain $I_{\Cc_1}=I_{\Cc_2}=(y,z)\subseteq\K[x,y,z]$. However, we have that $W_{\Cc_1}(x,1)=1+2x+x^2+4x^3+8x^4$ and $W_{\Cc_2}(x,1)=1+2x+3x^2+5x^3+5x^4$.
\end{example}
Determining whether it is possible to associate to an $R$-linear code a more complex ideal that allows to recover the weight enumerator remains an open problem. See Remark \ref{rmk:ideal_rk} for an obstacle to recovering the generalized weights of a rank-metric code from  an ideal associated to the latroid.

\section{Some families of $\F_q$-linear codes}

In this section, we discuss some interesting families of codes that can be studied effectively using our approach.

\subsection{Rank-metric codes}

We start by recalling the definition of $q$-polymatroid. Notice that when the function $\rho$ is integer-valued, the definition recovers that of $q$-matroid.  

\begin{definition}
A $q$-polymatroid is a pair $(\F_q^n,\rho)$ where $\rho:\mathcal{M}(\F_q^n)\rightarrow\mathbb{R}$ is a function such that: 
\begin{enumerate}
\item[P1.] $0\leq \rho(V)\leq\dim(V)$ for any $V\in\mathcal{M}(\F_q^n)$,
\item[P2.] $\rho(V_1)\leq\rho(V_2)$ for $V_1\leq V_2\in\mathcal{M}(\F_q^n)$,
\item[P3.] $\rho(V_1+V_2)+\rho(V_1\cap V_2)\leq \rho(V_1)+\rho(V_2)$ for $V_1,V_2\in\mathcal{M}(\F_q^n)$.
\end{enumerate}
\end{definition}

Let $(\F_q^n,\rho)$ be a $q$-polymatroid. Clearly, the dimension function is modular and strictly increasing, and $\dim(0)=\rho(0)=0$. Moreover, P2 implies that $0\leq \rho(V_2)-\rho(V_1)$ for $V_1\leq V_2$. On the other hand, since $\mathcal{M}(\F_q^n)$ is relatively complemented, there exists $V_3\leq V_2$ such that $V_1\cap V_3=0$, $V_1+V_3=V_2$. By the submodularity of $\rho$ we obtain $\rho(V_1)+\rho(V_3)\geq\rho(V_2)$, hence
$$\dim(V_2)-\dim(V_1)=\dim(V_3)\geq \rho(V_3)\geq\rho(V_2)-\rho(V_1).$$
Therefore, $\rho$ is bounded increasing with respect to the dimension. We conclude that any $q$-polymatroid can be regarded as an $\mathbb{R}$-latroid $(\rho,\dim,\mathcal{M}(\F_q^n))$. Conversely, it is clear that an $\mathbb{R}$-latroid $(\rho,\dim,\mathcal{M}(\F_q^n))$ is also a $q$-polymatroid.

Now we show how to associate a latroid to a rank-metric code $\Cc\in\F_q^{m\times n}$. We denote by $\mathrm{rowsp}(\Cc)$ the space generated by all rows of all matrices in $\Cc$ and by $\Cc(V)$ the largest subcode of $\Cc$, whose rowspace is contained in $V\in\mathcal{M}(\F_q^n)$. We define the function $\rho_{\Cc}:\mathcal{M}(\F_q^n)\rightarrow\mathbb{R}$ as
$$\rho_{\Cc}(V)=m\dim(V)-\dim(\Cc(V)).$$

\begin{proposition}\label{proposition:ranklatroid}
Let $\Cc\subseteq\F_q^{m\times n}$ be a rank-metric code. The triple $(\rho_{\Cc},m\dim,\mathcal{M}(\F_q^n))$ is a $\mathbb{Z}$-latroid.
\end{proposition}

\begin{proof}
The axioms L1, L2, and L3 are trivially satisfied. For every $V_1\leq V_2\in\mathcal{M}(\F_q^n)$ we have $\Cc(V_1)\subseteq\Cc(V_2)$, hence
$\rho_{\Cc}(V_2)-\rho_{\Cc}(V_1)\leq m(\dim(V_2)-\dim(V_1))$. Moreover, it is not hard to prove that $\dim(\Cc(V_2))\leq\dim(\Cc(V_1))+m(\dim(V_2)-\dim(V_1))$, which implies that $\rho_{\Cc}(V_2)-\rho_{\Cc}(V_1)\geq0$. Hence, L4 is satisfied. Since $\Cc(V_1)+\Cc(V_2)\subseteq\Cc(V_1+V_2)$ and $\Cc(V_1)\cap\Cc(V_2)=\Cc(V_1\cap V_2)$, we have that the function $\rho_{\Cc}$ is submodular. 
\end{proof}

\begin{remark}\label{rmk:polym_paper}
In~\cite{gorla2020rank} the authors associate to a rank-metric code $\Cc$ the $q$-polymatroid $(\tilde\rho_{\Cc},\F_q^n)$, where $\tilde\rho_{\Cc}:\mathcal{M}(\F_q^n)\rightarrow\mathbb{Q}$ is defined as $$\tilde\rho_{\Cc}(V)=\frac{\dim(\Cc)-\dim(C(V^{\perp}))}{m}.$$
Since $(\tilde\rho_{\Cc},\F_q^n)$ is a $q$-polymatroid, then $(\tilde\rho_{\Cc},\dim,\mathcal{M}(\F_q^n))$ is a $\mathbb{Q}$-latroid. Although the latroid $(\rho_{\Cc},\dim,\mathcal{M}(\F_q^n))$ of Proposition~\ref{proposition:ranklatroid} and $(\tilde\rho_{\Cc},\dim,\mathcal{M}(\F_q^n))$ are different latroids, they express the same information. In fact, $$\tilde\rho_{\Cc}(V)=\frac{\rho_{\Cc}(V^{\perp})-m\dim(V^{\perp})+\dim(C)}{m}.$$
However, we find $\rho_{\Cc}$ to be a more natural choice. For example, the independent elements of $(\rho_{\Cc},m\dim,\mathcal{M}(\F_q^n))$ are the spaces $V$ for which there are no elements of $\Cc$ whose rowspace is contained in $V$. Moreover, the circuits of $(\rho_{\Cc},m\dim,\mathcal{M}(\F_q^n))$ correspond to minimal supports in~$\Cc$.
\end{remark}

The next proposition shows that the generalized rank weights of a rank-metric code are determined by those of the latroid associated to it. This is not surprising, given Remark~\ref{rmk:polym_paper} and the fact that~\cite[Theorem~7.1]{gorla2020rank} shows that the generalized rank weights of a rank metric code $\Cc$ are determined by $(\tilde\rho_{\Cc},\dim,\mathcal{M}(\F_q^n))$.
We refer to~\cite[Section~11.5]{gorla2021rank} for the definition of generalized rank weights.

\begin{proposition}\label{proposition:generalizedrankweights}
Let $m\geq n$ be positive integers. If $m>n$, then the generalized rank weights of a rank-metric code $\Cc\subseteq \F_q^{m\times n}$ are equal to the generalized weights of the associated latroid of Proposition~\ref{proposition:ranklatroid} multiplied by $m$, i.e.,
\begin{equation*}
md_r(\Cc)=d_r(\rho_{\Cc},m\dim,\mathcal{M}(\F_q^n)).
\end{equation*}
If $m=n$, then 
\begin{equation*}
nd_r(\Cc)=\min\{d_r(\rho_{\Cc},n\dim,\mathcal{M}(\F_q^n)),d_r(\rho_{\Cc^t},n\dim,\mathcal{M}(\F_q^n))\}.
\end{equation*}
\end{proposition}

\begin{proof}
For $m>n$ we have that
\begin{equation*}
    \begin{split}
         d_r(\rho_{\Cc},m\dim,\mathcal{M}(\F_q^n))&=\min_{V\in\mathcal{M}(\F_q^n)}\{m\dim(L):\dim(\Cc(L))\geq r\}=\\&=\min\{\dim(\mathcal{A}): \mathcal{A}\text{ is an optimal anticode and}\dim(\Cc\cap A)\geq r\}=md_r(\Cc),
    \end{split}
    \end{equation*}
    where the equality in the middle is due to the fact that every optimal anticode $\mathcal{A}$ is uniquely determined by its rowspace and $\dim(\mathcal{A})=m\dim(\mathrm{rowsp}(\mathcal{A}))$. The proof for $m=n$ follows from similar arguments.
\end{proof}

\subsection{Sum-rank metric codes}

In~\cite[Definition 41]{panja2023some} the authors introduced the concept of sum matroid in order to associate a combinatorial object to $\F_{q^m}$-linear sum-rank metric codes. Using latroids, we can extend their ideas to arbitrary sum-rank metric codes.

Given a sum-rank metric code $\Cc\subseteq \prod_{i=1}^{\ell}\F_q^{m_i\times n_i}$, define $\rho_{\Cc}:\Ll\rightarrow \mathbb{R}$ as
$$\rho_{\Cc}(L)=\lVert L\rVert-\dim(\Cc(L)),$$
where $\Cc(L)$ is the set of codewords of $\Cc$ with columnspace contained in $L=(V_1,\dots,V_{\ell})$ and $\lVert L\rVert=\sum_{i+1}^{\ell}m_i\dim(V_i)$. Proceeding as in Proposition~\ref{proposition:ranklatroid}, one can prove the following.

\begin{proposition}\label{proposition:sumranklatroid}
    The triple $(\rho_{\Cc},\lVert\cdot\rVert,\Ll)$ is a $\Z$-latroid.
\end{proposition}

The next proposition makes the relation between the $\mathbb{Q}$-latroids of a sequence of rank metric codes and the $\Z$-latroid of their direct product explicit.

\begin{proposition}
    Let $\Cc\subseteq \prod_{i=1}^{\ell}\F_q^{m_i\times n_i}$ be a sum rank-metric code. If $\Cc=\prod_{i=1}^{\ell}\Cc_i$, then
    \begin{equation*}
        (\rho_{\Cc},\lVert\cdot\rVert,\Ll)=\bigoplus_{i=1}^{\ell}m_i(\rho_{\Cc_i},\dim,\mathcal{M}(\F_q^{n_i})),
    \end{equation*}
    where $m(\rho,\lVert\cdot\rVert,\Ll)$ denotes the direct dum of the latroid $(\rho,\lVert\cdot\rVert,\Ll)$ with itself $m$ times.
\end{proposition}

\begin{proof}
    If $\Cc=\prod_{i=1}^{\ell}\Cc_i$, then we have $\Cc(L)=\prod_{i=1}^{\ell}\Cc_i(L_i)$. So, we obtain $\rho_{\Cc}=\sum_{i=1}^{\ell}m_i\rho_{\Cc_i}$ and $\lVert\cdot\rVert=\sum_{i=1}^{\ell}m_i\dim(\cdot)$. We conclude by applying the definition of direct sum.
\end{proof}

The generalized weights of a sum-rank metric code are determined by the associated latroid, but not by the generalized weights of the associated latroid. In fact, given a sum-rank metric code $\Cc\subseteq \prod_{i=1}^{\ell}\F_q^{m_i\times n_i}$, we have
\begin{equation*}
    \begin{split}
        &d_r(\rho_{\Cc},\lVert\cdot\rVert,\Ll)=\{\lVert L\rVert:\dim(\Cc(L))\geq r\}=\\
        &=\{\dim(\mathcal{A}):\mathcal{A}=\mathcal{A}_1\times\dots\times \mathcal{A}_\ell \mbox{ where }\mathcal{A}_i\subseteq\F_q^{m_i\times n_i}\mbox{are o.a. and} \dim(\Cc \cap \mathcal{A})\geq r\},
    \end{split}
\end{equation*}
while following~\cite[Definition VI.1]{ourpaper} we obtain
$$d_r(\Cc)=\{\mathrm{maxsrk}(\mathcal{A}):\mathcal{A}=\mathcal{A}_1\times\dots\times \mathcal{A}_\ell \mbox{ where }\mathcal{A}_i\subseteq\F_q^{m_i\times n_i}\mbox{are o.a. and} \dim(\Cc \cap \mathcal{A})\geq r\}.$$
As often happens for sum-rank metric codes, we can prove an equality in the case where all the $m_i$'s are equal. The proof is analogous to that of Proposition~\ref{proposition:generalizedrankweights}.

\begin{proposition}
    If $m=m_1=\dots=m_{\ell}$ and $m_i>n_i$ for all $i\in[\ell]$, then the generalized sum-rank weights of a sum-rank metric code $\Cc\subseteq \prod_{i=1}^{\ell}\F_q^{m_i\times n_i}$ are equal to the generalized weights of the associated latroid of Proposition~\ref{proposition:sumranklatroid} multiplied by $m$, i.e.,
    \begin{equation*}
        d_r(\rho_{\Cc},\lVert\cdot\rVert,\Ll)=md_r(\Cc).
    \end{equation*}
\end{proposition}

\bibliographystyle{plain}

\end{document}